\documentclass[11pt,reqno]{article}
\usepackage{amsthm,amsmath, amssymb, amsopn, amsfonts}
\usepackage[margin=1in]{geometry}
\usepackage[colorlinks,citecolor=blue,urlcolor=blue]{hyperref}
\usepackage{graphicx}
\usepackage{tikz}
\usetikzlibrary{decorations.fractals}
\usetikzlibrary{patterns,arrows,shapes,decorations.pathreplacing}
\usepackage{pgfplots}
\usepackage{stmaryrd} 
\usepackage{soul,xcolor}
\setstcolor{red}

\numberwithin{equation}{section}
\theoremstyle{plain}
\newtheorem{Definition}{Definition}[section]
\newtheorem{Remark}{Remark}[section]
\newtheorem{Theorem}{Theorem}[section]
\newtheorem{Lemma}{Lemma}[section]
\newtheorem{Proposition}{Proposition}[section]
\newtheorem{Corollary}{Corollary}[section]
\newtheorem{Assumption}{Assumption}[section]

\newcommand{\bee}{\begin{equation*}}
\newcommand{\eee}{\end{equation*}}
\newcommand{\be}{\begin{equation}}
\newcommand{\ee}{\end{equation}}
\newcommand{\bi}{\begin{itemize}}
\newcommand{\ei}{\end{itemize}}
\def \E{\mathbb{E}}
\def \F{\mathbb{F}}
\def \N{\mathbb{N}}
\def \P{\mathbb{P}}
\def \Q{\mathbb{Q}}
\def \R{\mathbb{R}}

\def \B{{\mathcal B}}

\def \Ec{{\mathcal E}}

\def \tEc{{\widetilde{\mathcal E}}}
\def \Fc{{\mathcal F}}
\def \Pc{{\mathcal P}}
\def \Lc{{\mathcal L}}

\def \Rb{{\overline R}}

\def \eps{\varepsilon}

\title{Optimal Equilibria for Multi-dimensional  Time-inconsistent Stopping Problems\thanks{We would like to thank Tiziano De Angelis for introducing the fine topology to the first author. This led to the initiation of this project.}}
\author{Yu-Jui Huang\thanks{
University of Colorado, Department of Applied Mathematics, Boulder, CO 80309-0526, USA, email: \texttt{yujui.huang@colorado.edu}. Partially supported by National Science Foundation (DMS-1715439) and the University of Colorado (11003573).}
 \and Zhenhua Wang\thanks{
University of Michigan, Department of Mathematics, Ann Arbor, MI 48109-1043, USA, email: \texttt{zhenhuaw@umich.edu}.}
}

\begin{document}
\maketitle

\begin{abstract}
We study an optimal stopping problem under non-exponential discounting, where the state process is a {\it multi-dimensional} continuous strong Markov process. The discount function is taken to be log sub-additive, capturing {\it decreasing impatience} in behavioral economics. On strength of probabilistic potential theory, we establish the existence of an optimal equilibrium among a sufficiently large collection of equilibria, consisting of finely closed equilibria satisfying a boundary condition. 
This generalizes the existence of optimal equilibria for one-dimensional stopping problems in prior literature.  
\end{abstract}

\smallskip



\section{Introduction}
Whereas {\it consistent planning} in Strotz \cite{Strotz55} has been widely applied to time-inconsistent control and stopping problems, it gains little recognition as a two-phase procedure. First, an agent should figure out strategies that he will actually follow over time, the so-called {\it equilibria} in the literature {\it (Phase I)}. Then, the agent needs to, according to Strotz \cite[p.173]{Strotz55}, ``find the best plan among those he will actually follow'' {\it (Phase II)}.

Thanks to the continuous-time formulation initiated in Ekeland and Lazrak \cite{EL06}, there has been vibrant research on time-inconsistent problems in the communities of stochastic control and mathematical finance; see e.g. \cite{ekeland2008investment}, \cite{bjork2017time}, \cite{HJZ12}, \cite{bjork2014mean}, \cite{ekeland2012time}, \cite{yong2012time}, among many others. Note that most of the developments were focused on \textit{Phase I} of consistent planning. To the best of our knowledge, in-depth investigation of \textit{Phase II} was initiated fairly recently in Huang and Zhou \cite{HZ19, huang2017optimal}. 

The notion of an {\it optimal equilibrium} is proposed for the first time in \cite{HZ19}: an equilibrium is optimal if it generates a larger value than any other equilibrium does, {\it everywhere} in the state space. This seems a fairly strong optimality criterion, as it requires a subgame perfect Nash equilibrium to be dominant on the entire state space. Nonetheless, \cite{HZ19} derives an optimal equilibrium for a discrete-time stopping problem under non-exponential discounting. Corresponding results in continuous time have been established in \cite{huang2017optimal}, relying on a detailed analysis of one-dimensional diffusions. Specifically, a key assumption in \cite{huang2017optimal} is that for any initial state $x\in\R$, the state process $X$ satisfies
\begin{equation}\label{1-d condition}
\P^x[\overline X_t > x] = \P^x[\underline X_t < x] =1\quad \hbox{for all $t> 0$}, 
\end{equation}
where $\overline X_t:= \max_{s\in[0,t]} X_s$ (resp. $\underline X_t := \min_{s\in[0,t]} X_s$) is the running maximum (resp. minimum) of $X$. This condition ensures that $X$ is ``diffusive'' enough, so that whenever $X$ reaches the boundary of a Borel subset $R$ of $\R$, it enters $R$ immediately. This allows us to focus on only {\it closed} equilibrium,  with no need to deal with equilibria of possibly pathological forms. Note that \eqref{1-d condition} is not restrictive in the one-dimensional case (i.e. $d=1$): any {\it regular} diffusion (in the sense of \cite[Definition V.45.2]{RW-book-00}) readily fulfills \eqref{1-d condition}, as recently observed in \cite[Remark 3.1]{HY19}.     

There is, however, no natural extension of \cite{huang2017optimal} to higher dimensions (i.e. $d\ge 2$). Inherently one-dimensional, \eqref{1-d condition} and its related consequences are easily violated by multi-dimensional diffusions; see the discussion below Proposition~\ref{prop:1-d result} for details. 

This paper aims at establishing the existence of an optimal equilibrium for a multi-dimensional stopping problem under non-exponential discounting. With no convenient regularity condition, such as \eqref{1-d condition}, in higher dimensions, we take a very different approach 
on strength of probabilistic potential theory. First, in terms of regularity of the state process, we assume only that a reference measure for $X$ exists (Assumption~\ref{asm:reference measure}). This serves as a minimal condition to ensure Borel measurability of involved stopping policies; see Corollary~\ref{coro:Theta(R) Borel}. Since this assumption is satisfied as long as a transition density of $X$ exists (Lemma~\ref{lem:density}), Corollary~\ref{coro:Theta(R) Borel} covers and generalizes all measurability results in \cite{huang2018time}, \cite{huang2017stopping}, and \cite{huang2017optimal}. Next, we set out to devise 
a sufficiently large collection $\Lc$ of stopping policies, among which we will find an optimal equilibrium. The flexibility that $\Lc$ is not required to contain all Borel stopping policies is important: 
focusing on more amenable and practically relevant stopping policies will facilitate the search for an optimal equilibrium (optimal among this class $\Lc$). 
In a sense, for the case $d=1$, tractability comes from enhanced regularity of $X$, i.e. \eqref{1-d condition}; for $d\ge 2$, it comes from additional structures of stopping policies. To construct the collection $\Lc$, we first restrict our attention to {\it finely closed} stopping policies, and then require the difference between them and their Euclidean closures to be sufficiently small, or more precisely, {\it semipolar}. The first restriction is without loss of any generality: as shown in Proposition \ref{prop.equiv.finely}, the fine closure of an equilibrium remains an equilibrium with the same values. The second restriction, on the other hand, is to exclude  stopping policies with somewhat pathological and practically irrelevant forms; see Remark~\ref{rem:exclude Q}. This class $\Lc$ already includes all closed stopping policies, and is closed under finite unions and countable intersections; see Lemma~\ref{lemma.set.0}. 

Our goal is to find an equilibrium optimal among $\tEc := \Ec\cap\Lc$, where $\Ec$ denotes the set of all equilibria. By invoking {\it Hunt's hypothesis} (Assumption~\ref{assum.hunt}), which states that {\it every semipolar set is polar}, we obtain two important consequences. First, given a stopping policy $R$, if we perform on $R$ one round of the fixed-point iteration in \cite{huang2018time} and get a smaller policy, this smaller one has to be an equilibrium; see Proposition~\ref{prop.decrease.theta}. Second, for any $R\in\tEc\subseteq \Lc$, since the difference between $R$ and $\overline R$, its Euclidean closure, is now polar (i.e. inaccessible to the process $X$), the first hitting time to $R$, denoted by $\rho_R$, becomes much more tractable as it must coincide with $\rho_{\overline R}$. Based on all this, a machinery for improving equilibria in $\tEc$ is developed: for any $R$, $T\in\tEc$, there exists another equilibrium in $\tEc$, contained in $R\cap T$, which generates larger values than both $R$ and $T$; see Proposition~\ref{prop.intersection}. By carrying out this machinery recursively, we construct an equilibrium $R_\star$ that is optimal among $\tEc$, with $\overline{R_\star} = \bigcap_{R\in\tEc} \overline R$; see Theorem~\ref{thm.optimal}, the main result of this paper. 

Given that $\rho_R=\rho_{\overline R}$ for all $R\in\tEc$, it is tempting to believe that we could restrict our attention from $\tEc$ to $\overline \Ec:= \{R\in\Ec: R\ \hbox{is closed}\}$, without loss of generality. This would make our results completely in line with the one-dimensional analysis in \cite{huang2017optimal}, which is built upon $\overline\Ec$. This is however {\it not} the case, as the relation ``$R\in\tEc$ if and only if $\overline R\in\Ec$'' does not hold in general; see Remark~\ref{rem:to closedE} for explanations, and Section~\ref{sec:eg} for an explicit example that demonstrates $R\in\tEc$ but $\overline R\notin \Ec$. 

\subsection{Related Recent Developments}
Lately, marked progress has been made in applying Stortz' equilibrium idea to time-inconsistent stopping problems in continuous-time diffusion models. It can be roughly categorized into two different kinds of definition of an equilibrium. The first one is based on the fixed-point iterative approach introduced in Huang and Nguyen-Huu \cite{huang2018time}, and further developed in \cite{huang2017stopping}, \cite{huang2017optimal}, and \cite{HY19}. The second definition adapts the stochastic control formulation in \cite{EL06} to the stopping context; this includes Ebert, Wei, and Zhou \cite{EWZ18} and Christensen and Lindensj\"{o} \cite{CL18, CL20}. A large part of these developments, notably, rely on the one-dimensional structure of the state process. 
The arguments in this paper can potentially shed new light on extending previous one-dimensional results to higher dimensions. 

Recently, in a continuous-time Markov chain model, Bayraktar, Zhang, and Zhou \cite{BZZ19} introduced a third kind of definition of an equilibrium for optimal stopping, based on the notion ``strong equilibrium'' introduced in Huang and Zhou \cite{HZ19-MOR} for stochastic control. A detailed analysis on the above three kinds of definition is performed in \cite{BZZ19}, particularly showing that an optimal equilibrium of the first kind (i.e., the same notion of an optimal equilibrium as in this paper) is in fact an equilibrium of the third kind. It is of great interest as future research to extend this finding to a multi-dimensional diffusion model, like the one in this paper.
For a detailed discussion on different definitions of an equilibrium in the more general stochastic control and stopping context, one may refer to \cite[Section 3.1]{hernandez2020me}.

In discrete-time models, two recent works give new results on {\it Phase II} of consistent planning. 
The existence of optimal equilibria for a time-inconsistent dividend problem has been analyzed in detail by Jin and Zhou \cite{JZ20}. 
Through a so-called ``myopic adjustment" procedure, \cite{christensen2019time} introduces the stability of equilibria, providing an alternative path to go beyond {\it Phase I} of consistent planning: among all equilibria, realistic decision makers choose the ones that are stable.  

\subsection{Organization of the Paper}
Section \ref{sec.modelsetup} introduces the time-inconsistent stopping problem, a new measurability result, and the idea of optimal equilibria. 
Section~\ref{sec:Lc} devises a suitable set of stopping policies for the rest of the paper to focus on. This set encodes desirable properties from probabilistic potential theory, yet remains general enough to cover practically relevant stopping policies. Under Hunt's hypothesis, Section~\ref{sec:optimalE} develops a machinery to improve equilibria, by which an optimal equilibrium is constructed. Section~\ref{sec:eg} presents an example to demonstrate explicitly that the Euclidean closure of a finely closed equilibrium need not be an equilibrium.


\section{The Setup and Preliminaries}\label{sec.modelsetup}
Let $\B$ be the Borel $\sigma$-algebra of $\R^d$, and $\Pc$ be the set of all probability measures on $(\R^d, \B)$. 
Consider an $\R^d$-valued time-homogeneous continuous strong Markov process $(X_t)_{t\geq 0}$. Let the collection $(P_t: \R^d\times \B\to [0,1])_{t\ge 0}$ denote the {\it transition function} of $X$. 
If there exist a collection $(p_t:\R^d\times\R^d\to \R_+)_{t\ge 0}$ of Borel measurable functions and a measure $\lambda$ on $(\R^d,\B)$ such that
\begin{equation}\label{density}
P_t(x,A) = \int_A p_t(x,y)\lambda(dy),\quad \forall t\ge0,\ x\in\R^d,\ \hbox{and}\ A\in\B, 
\end{equation}
we call $(p_t)_{t\ge 0}$ a {\it transition density} of $X$ with respect to $\lambda$.

On the path space $\Omega:= C([0,\infty);\R^d)$, the set of continuous functions mapping $[0,\infty)$ to $\R^d$, let $\mathbb F^X=\{\Fc^X_t\}_{t\ge 0}$ be the raw filtration generated by $X$. For each $\mu\in \Pc$, we denote by $\B^\mu$ the completion of $\B$ by $\mu$ and by $\mathbb F^\mu =\{\Fc_t^\mu\}_{t\ge 0}$ the $\mu$-augmentation of $\mathbb F^X$. We further consider the universal filtration $\mathbb F=\{\Fc_t\}_{t\ge 0}$ defined by
$\Fc_t:=\bigcap_{\mu\in \Pc} \Fc^\mu_t$ for all $0\leq t< \infty$, and denote by $\mathcal T$ the set of $\F$-stopping times. For any $x\in\R^d$, let $X^x$ denote the process $X$ with initial value $X_0=x$. Moreover, the probability measure on $(\Omega,\Fc^X_\infty)$ generated $X^x$ (i.e. the law of $(X^x_t)_{t\geq 0}$) is denoted by $\P^x$ 
and the expectation taken under $\P^x$ is denoted by $\E^x$. 


Consider a payoff function $f: \R^d\rightarrow \R_+$, assumed to be continuous. Also consider a discount function $\delta: [0, \infty)\rightarrow [0,1]$, assumed to be continuous, nonincreasing, and satisfy $\delta(0)=1$ and $\lim_{t\to\infty}\delta(t)=0$.  A classical optimal stopping problem is formulated as
\begin{equation}\label{classical}
\sup_{\tau\in\mathcal T} \E^x[\delta(\tau)f(X_{\tau})].
\end{equation}
It is well-known (see e.g. \cite[Appendix D]{KS-book-98} and \cite{Shiryayev-book-78}) that under fairly general conditions, for any initial state $x\in\R^d$, there exists an optimal stopping time $\widetilde\tau_x\in \mathcal T$ for \eqref{classical}. However, as long as the discount function is {\it not} of the exponential form $\delta(t) := e^{-\alpha t}$ with $\alpha>0$, the problem \eqref{classical} is in general {\it time-inconsistent}. That is, optimal stopping times obtained at different moments, such as $\widetilde\tau_x$ at time $0$ and $\widetilde\tau_{X^x_t}$ at time $t>0$, may not be consistent with each other. This is problematic: even if a maximizer $\widetilde\tau_x$ of \eqref{classical} can be found, our future self at any time $t>0$ is tempted to employ $\widetilde\tau_{X^x_t}$, optimal for him at time $t$, rather than stick with $\widetilde\tau_x$; see \cite[Section 2]{huang2018time} for a detailed demonstration of such time inconsistency. 


Throughout this paper, we will assume that the discount function $\delta$ satisfies 
\begin{equation}\label{DI}
\delta(s)\delta(t)\leq\delta(s+t),\quad\forall\,s,t>0.
\end{equation}
This covers a wide range of non-exponential discount functions from empirical studies; see the discussion below \cite[Assumption 3.12]{huang2018time}. In economic terms, \eqref{DI} captures {\it decreasing impatience}, the fact that people discount more steeply over time intervals closer to the present. This feature of empirical discounting is well-documented in behavioral economics; see e.g. \cite{Thaler81}, \cite{LT89}, and \cite{LP92}.


Under \eqref{DI}, time inconsistency is a genuine problem. Strotz \cite{Strotz55} proposes {\it consistent planning} as a solution: an agent should take into account his future selves' disobedience, and find a strategy that once being enforced, none of his future selves would deviate from. Such strategies, called {\it equilibiria} in the literature, can be formulated using the game-theoretic framework initiated in \cite{huang2018time}. 


\subsection{A Game-theoretic Framework}
Under our time-homogeneous Markovian setup, we will focus on first entry times to Borel subsets of $\R^d$, i.e. $\tau_R:=\inf\{t\ge 0: X_t\in R\}$ for $R\in\B$, instead of all general stopping times. The involved region $R\in\B$ will be called a {\it stopping policy} constantly. 

Suppose that an agent initially planned to take $R\in\B$ as his stopping policy. Now, at any state $x\in\R^d$, the agent carries out the game-theoretic reasoning: ``assuming that all my future selves will follow $R\in\B$, what is the best stopping strategy today in response to that?'' To this end, the agent compares the payoff of immediate stopping, i.e. $f(x)$, and the payoff of continuation, i.e. 
\begin{equation}\label{J}
J(x, R):=\E^x[\delta(\rho_R)f(X_{\rho_R})], 
\end{equation}
where $\rho_R$ is the first hitting time of $X$ to $R$ defined by 
\begin{equation}\label{rho_R}
\rho_R:=\inf\{t>0: X_t\in R\}.
\end{equation}
As explained in detail in \cite[Section 2.1]{huang2017optimal} (see also \cite[Section 2]{huang2017stopping} or \cite[Section 3.1]{huang2018time}), the best stopping strategy for the agent at $x\in\R^d$ is the first entry time to the region
\begin{equation}\label{Theta}
\Theta(R):=S(R)\cup (I(R)\cap R),
\end{equation}
where
\begin{equation}\label{S, I, C}
\begin{aligned}
S(R) &:=\{x\in \R^d : J(x, R)<f(x)\},\\
I(R)&:=\{x\in \R^d: J(x, R)=f(x)\},\\
C(R)&:=\{x\in \R^d: J(x, R)>f(x)\}.
\end{aligned}
\end{equation} 

\begin{Remark}
In \eqref{J}, we need specifically the first hitting time $\rho_R$, instead of the first entry time $\tau_R=\inf\{t\ge 0: X_t\in R\}$. This was explained in the discussion below \cite[(3.5)]{huang2018time}. 
\end{Remark}

Appropriate conditions are needed to make the above formulation mathematically rigorous. First, to ensure that \eqref{J} is well-defined, we impose the integrability condition: for any $x\in\R^d$, 
\begin{equation}\label{assum.bound}
\E^x\bigg[\sup_{t\in [0,\infty]}\delta(t)f(X_t)\bigg]<\infty, 
\end{equation}
where we take $\delta(\infty)f(X^x_\infty):=\limsup_{t\to\infty} \delta(t)f(X^x_t)$, which is in line with \cite[Appendix D]{KS-book-98}. 
Second, to ensure that $\Theta(R)$ in \eqref{Theta} is indeed a stopping policy, i.e. $\Theta(R)\in\B$, 
certain regularity of $X$ is required, which will be investigated closely below. 


\subsection{Measurability}
Let us first recall the following concepts from probabilistic potential theory. 

\begin{Definition}[\cite{blumenthal2007markov}, Definition II.3.1]
For any $R\in \B$ and $x\in \R^d$, the {\it potential of} $R$ is defined by 
\bee
U(x,R):=\int_0^\infty P_t(x, R)dt.
\eee 
We say $R\in \B$ is of zero potential if $U(x,R)=0$ for all $x\in\R^d$.
\end{Definition}

\begin{Definition}[\cite{blumenthal2007markov}, Definition V.1.1]
A measure $\lambda$ on $\R^d$ is called a {\it reference measure} for $X$, provided that it is a countable sum of finite measures such that $R\in \B$ is of zero potential if and only if $\lambda(R)=0$. 
\end{Definition}

Let us introduce the first main assumption of this paper.

\begin{Assumption}\label{asm:reference measure}
A {reference measure} for $X$ exists. 
\end{Assumption}

A convenient sufficient condition of Assumption~\ref{asm:reference measure} is provided in the next result, whose proof is relegated to Appendix~\ref{sec:proof of lem:density}. 

\begin{Lemma}\label{lem:density}
Suppose that $X$ has a transition density with respect to a measure $\lambda$ on $\R^d$. If $\lambda$ is a countable sum of finite measures, then for any $\alpha>0$, $\lambda^{\alpha}:\B\to\R_+$ defined by 
\[
\lambda^{\alpha}(R):= \int_{\R^d}\left(\int_0^\infty e^{-\alpha t} P_t(x,R)dt \right)\lambda(dx)
\]
is a reference measure for $X$.
\end{Lemma}


Lemma~\ref{lem:density} indicates that Assumption~\ref{asm:reference measure} is not restrictive for a wide range of applications, as a large class of diffusion processes have a transition density with respect to the Lebesgue measure in $\R^d$. For the $d$-dimensional Brownian motion $B$, such a transition density is $p_t(x, y)= (2\pi t)^{-d/2}\exp(-\frac{|x-y|^2}{2t})$; see \cite[Exercise 5.6.17]{karatzas1998brownian}. For an It\^{o} diffusion given by
\begin{equation}\label{diffusion}
dX_t = b(X_t)dt + \sigma(X_t)dB_t, 
\end{equation}
as long as there exists a weak solution unique in distribution, and the coefficients $b$ and $\sigma$ are continuous and grow at most linearly, the transition density of $X$ with respect to the Lebesgue measure exists and is characterized by the fundamental solution to a Cauchy problem; see \cite[p.369]{karatzas1998brownian}. For general Markov processes, we refer readers to  \cite[Proposition V.1.2]{blumenthal2007markov} for a general sufficient condition of the existence of a reference measure. 

When a reference measure exists, we have the following handy approximation for hitting times; see Proposition 10 in \cite[Section 3.5]{chung2006markov} and \cite[Exercise V.1.20]{blumenthal2007markov}.

\begin{Lemma}\label{lemma.hypothesis.l}
Suppose Assumption~\ref{asm:reference measure} holds. For any $R\in \B$, there exist a nondecreasing sequence of compact sets $K_n\subseteq R$ such that $\rho_{K_n}\to\rho_R$ $\P^x$-a.s. for all $x\in\R^d$. 
\end{Lemma}

Borel measurability of the map $x\mapsto J(x, R)$ can then be established. 

\begin{Proposition}\label{prop:J Borel} 
\begin{itemize}
\item [(i)] For any closed $R\in\B$, $x\mapsto J(x, R)$ is Borel measurable. 
\item [(ii)] Assume Assumption~\ref{asm:reference measure} and that \eqref{assum.bound} holds for all $x\in\R^d$. 
Then, for any $R\in \B$, $x\mapsto J(x, R)$ is Borel measurable.   
\end{itemize}
\end{Proposition}

\begin{proof}
(i) For any $R\in\B$ that is closed, $\rho_R$ is an $\mathbb F^X$-stopping time, thanks to \cite[Problem I.2.7]{karatzas1998brownian}. Hence, the random variable $H(\omega):= \delta(\rho_R(\omega))f(X_{\rho_R}(\omega))$ is $\Fc^X_\infty$-measurable. By \cite[Theorem I.3.6]{blumenthal2007markov}, $x\mapsto\E^x[H]=J(x,R)$ is Borel measurable.

(ii) For any $R\in\B$, Lemma \ref{lemma.hypothesis.l} asserts the existence of compact sets $K_n\subseteq R$ such that $\rho_{K_n}\to\rho_R$ $\P^x$-a.s. for all $x\in\R^d$. Thanks to this, the continuity of $\delta$, $f$, and $t\mapsto X_t$, and \eqref{assum.bound}, we conclude from the dominated convergence theorem that 
\begin{equation}\label{J_n to J}
J(x, K_n)=\E^x[\delta(\rho_{K_n})f(X_{\rho_{K_n}})]\to \E^x[\delta(\rho_R)f(X_{\rho_R})]= J(x,R),\quad \forall x\in \R^d.
\end{equation}
By part (a), $x\mapsto J(x,K_n)$ is Borel measurable for all $n\in\N$. Hence, $J(x, R)$ is Borel measurable in view of \eqref{J_n to J}.
\end{proof}

\begin{Corollary}\label{coro:Theta(R) Borel}
Assume Assumption~\ref{asm:reference measure} and that \eqref{assum.bound} holds for all $x\in\R^d$. 
For any $R\in\B$, we have $\Theta(R)\in\B$. 
\end{Corollary}

\begin{proof}
For any $R\in\B$, by Proposition~\ref{prop:J Borel} (ii), $x\mapsto J(x,R)$ is Borel measurable. Hence, by definition $S(R)$, $I(R)$, and $C(R)$ all belong to $\B$. It follows that $\Theta(R)=S(R)\cup ( I(R)\cap R)\in \B$.
\end{proof}

By Corollary~\ref{coro:Theta(R) Borel}, $\Theta$ can be viewed as an operator acting on $\B$, i.e. $\Theta:\B\to\B$. An equilibrium can then be formulated as a fixed point of $\Theta$.

\begin{Definition}\label{def.equilibrium}
$R\in\B$ is called an equilibrium if $\Theta(R)=R$. We denote by $\Ec$ the set of all equilibria.
\end{Definition}
It can be checked directly that the entire space $R=\R^d$ is an equilibrium. Moreover, a large number of (or even all) equilibria can be found, by the fixed-point iteration introduced in \cite{huang2018time}: one starts with an arbitrary $R\in\B$, and apply $\Theta$ to it repetitively until an equilibrium is reached; see also Remark~\ref{rem:convergence}.


\subsection{Optimal Equilibria}
Finding equilibria, however, is only the first phase of {\it consistent planning} in Strotz \cite{Strotz55}. In the second phase, the agent should choose the {\it best} one among all equilibria. This has not been studied in the literature, except in \cite{HZ19} and \cite{huang2017optimal}. Following \cite[Section 2.2]{huang2017optimal}, for each $R\in\Ec$, we define the associated value function by
\[
V(x,R) := f(x)\vee J(x,R)\quad \forall x\in\R^d.
\]
\begin{Definition}
Given any $\Ec'\subseteq\Ec$, a set $R\in\Ec'$ is called an optimal equilibrium among $\Ec'$ if for any other $T\in\Ec'$,
 $V(x, R)\geq V(x, T)$ for all $x\in \R^d$.
\end{Definition}
  
In the one-dimensional case (i.e. $d=1$), the existence of an optimal equilibrium among the entire set $\Ec$ is established in \cite{huang2017optimal}. 

\begin{Proposition}[\cite{huang2017optimal}, Theorem 4.1]\label{prop:1-d result}
Suppose that for any $x\in\R$, we have \eqref{1-d condition}, 
\eqref{assum.bound}, and 
\begin{equation}\label{old asm}
\delta(t)f(X^x_t)\rightarrow 0\ \ \text{as}\ \ t\rightarrow \infty\quad \P^x\hbox{-a.s.}
\end{equation} 
Then, the set
\begin{equation}\label{1-d fR}
{R_\star} := \bigcap_{R\in\Ec,\ R\ \text{closed}} R
\end{equation}
is an optimal equilibrium among $\Ec$.
\end{Proposition}
  
As mentioned in the introduction, \eqref{1-d condition} ensures that whenever $X$ reaches the boundary of a Borel subset $R$ of $\R$, it enters $R$ immediately. This allows us to focus on only {\it closed} equilibrium, as indicated by \eqref{1-d fR}. 
As shown in \cite[Lemma 3.1]{huang2017stopping}, \eqref{1-d condition} is satisfied by a large class of one-dimensional It\^{o} diffusions. Even more generally, any {\it regular} diffusion (in the sense of \cite[Definition V.45.2]{RW-book-00}) fulfills \eqref{1-d condition}, as recently observed in \cite[Remark 3.1]{HY19}.     

Proposition~\ref{prop:1-d result} does {\it not} naturally extend to higher dimensions. First, due to the involved $\overline X$ and $\underline X$ processes, \eqref{1-d condition} is inherently one-dimensional, with no natural extension in higher dimensions. Moreover, the proof of Proposition~\ref{prop:1-d result} relies crucially on a consequence of \eqref{1-d condition}: $\rho_{\{x\}}=0$ $\P^x$-a.s. for all $x\in\R$, i.e. the process $X$ re-visits its initial point immediately. This condition is mostly violated in higher dimensions. For instance, when $X$ is a $d$-dimensional Brownian motion, $\rho_{\{x\}}=\infty$ $\P^x$-a.s. for all $x\in\R^d$, whenever $d\ge 2$. 

The goal of this paper is to establish a multi-dimensional counterpart of Proposition~\ref{prop:1-d result}. With no convenient regularity condition, such as \eqref{1-d condition}, to rely on, we will look for an optimal equilibrium among a sufficiently large subset $\Ec'\subseteq\Ec$, which is allowed to be properly contained in $\Ec$. This flexibility is important:
focusing on more amenable and practically relevant stopping policies will facilitate the search for an optimal equilibrium. 
In a sense, for $d=1$, tractability comes from desirable regularity of $X$, i.e. \eqref{1-d condition}; for $d\ge 2$, it will come from additional structures of stopping policies. The search for an appropriate subset $\Ec'\subseteq \Ec$, which needs to be amenable enough but still sufficiently large, will  be the focus of the next section. 

\begin{Remark}
In the time-consistent case of exponential discounting (i.e., $\delta(s)\delta(t)=\delta(s+t)$ $\forall s,t\ge 0$), the standard optimal stopping time defined by
\[
\hat\tau : = \inf\{t\ge 0 : f(X_t) = U(X_t)\}\in\mathcal T,\quad \hbox{with}\quad U(x):= \sup_{\tau\in\mathcal T}\E^x[\delta(\tau) f(X_\tau)]\ge f(x),
\]
readily yields an optimal equilibrium. 
To see this, note that \cite[Proposition 3.11]{huang2018time} already shows (in a more general time-inhomogeneous setting) that 
$\hat R := \{x\in\R^d: f(x) = U(x)\}$ is an equilibrium. Then, by the definitions of $\hat R$ and $\hat\tau$, we have
\[
f(x)\vee J(x,\hat R) = \E^x[\delta(\hat \tau)f(X_{\hat\tau})] = \sup_{\tau\in\mathcal T}\E^x[\delta(\tau) f(X_\tau)]\ge f(x)\vee J(x,R),\quad \forall R\in\Ec,
\]
where the second equality follows from $\hat\tau$ being an optimal stopping time. Hence, we conclude that $\hat R$, the so-called optimal stopping region in the literature, is an optimal equilibrium. 
\end{Remark}


\section{Finely Closed Equilibria with a Boundary Condition}\label{sec:Lc}
On strength of probabilistic potential theory, a suitable subset of $\Ec$ to focus on will be devised in this section. We will first restrict our attention to {\it finely closed} stopping policies (see Definition~\ref{def.regular} below), and then further require the difference between them and their Euclidean closures to be sufficiently small. As we will see, the first restriction is without loss of any generality, while the second is to exclude equilibria with possibly pathological and practically irrelevant forms. 

\subsection{Finely Closed Stopping Policies}\label{subsec:finely closed}
Let us recall several essential concepts from probabilistic potential theory. 
\begin{Definition}\label{def.regular}
Given $R\in \B$, a point $x\in\R^d$ is said to be regular to $R$ if $\rho_R=0$ $\P^x$-a.s. The set of all regular points to $R$ (with respect to $X$) is denoted by $R^r$, and we call 
\begin{equation}\label{A^*}
R^*: =R\cup R^r.
\end{equation}  
the fine closure of $R$. In addition, $R$ is said to be finely closed if $R=R^*$. 
\end{Definition} 

\begin{Remark}
By Blumenthal's zero-one law (Theorem 6 in \cite[Section 2.3]{chung2006markov}), for any $x\in\R^d$ and $R\in\B$, $\P^x(\rho_R =0)$ is either 0 or $1$. Hence, if $x\in\R^d$ is not regular to $R$, then $\rho_R>0$ $\P^x$-a.s.
\end{Remark}

\begin{Remark}\label{rem:fine topology}
Adding to a set all its regular points, as in \eqref{A^*}, is the closure operation under the fine topology (see e.g. \cite[p.107]{chung2006markov}). Hence, for any $R\in\B$, 
\begin{equation}\label{A^**=A^*}
(R^*)^*=R^*,\quad \hbox{or}\quad (R^*)^r\subseteq R^*.
\end{equation} 
\end{Remark}

\begin{Remark}\label{rem:X_rhoA in A^*}
For any $R\in\B$, 
\begin{equation}\label{X_rhoA in A^*}
X_{\rho_R}\in R^*\ \ \P^x\hbox{-a.s. on}\ \{\rho_R<\infty\},\quad \forall x\in\R^d.
\end{equation}
Indeed, for $\P^x$-a.e. $\omega\in\{\rho_R<\infty\}$, if $X_{\rho_R}(\omega)\notin R$, by the definition of $\rho_R$, $X_{\rho_R}(\omega)$ must be regular to $R$, i.e. $X_{\rho_R}(\omega)\in R^r$. Hence, $X_{\rho_R}(\omega)\in R\cup R^r=R^*$.  
\end{Remark}

\begin{Remark}\label{rem:A^r in I}
Fix $R\in\B$. For any $x\in R^r$, as $\rho_R=0$ $\P^x$-a.s., $J(x, R)=f(x)$. Hence, 
\begin{equation}\label{A^r in I}
R^r\subseteq I(R).
\end{equation}
\end{Remark}

Borel measurability of $R^r$ and $R^*$ can be established under Assumption~\ref{asm:reference measure}. 

\begin{Corollary}
Assume Assumption~\ref{asm:reference measure} and that \eqref{assum.bound} holds for all $x\in\R^d$. Then, for any $R\in\B$, $R^r\in\B$ and thus $R^*\in\B$.  
\end{Corollary}

\begin{proof}
For any  $R\in \B$, by the same arguments as in Proposition~\ref{prop:J Borel} (ii), with $J(x,R)$ replaced by $\E^x[e^{-\rho_R}]$, we can show that $x\mapsto \E^x[e^{-\rho_R}]$ is Borel measurable. Thus, $R^r=\{x\in \R^d: \E^x[e^{-\rho_R}]=1\}\in \B$. It follows that $R^*=R\cup R^r\in\B$.  
\end{proof}

A key observation is that first hitting times to $R$ and to $R^*$ must coincide. 

\begin{Lemma}\label{lemma.equiv}
For any $R\in \B$, $\rho_R = \rho_{R^*}$ $\P^x${-a.s.} for all $x \in \R^d$. 
Hence,
\begin{align}
&\hspace{0.6in}J(x, R)=J(x,R^*)\quad \forall x\in \R^d,\label{lemma.finely2} \\  
&S(R)=S(R^*),\ \ I(R)=I(R^*),\ \ C(R)=C(R^*).\label{lemma.finely2'}
\end{align}
\end{Lemma}

\begin{proof}
Fix $R\in\B$ and $x\in\R^d$. Since $R\subseteq R^*$, $\rho_{R^*}\leq \rho_R$. Assume $\rho_{R^*}<\infty$, otherwise $\rho_{R^*}=\rho_R$ trivially holds. By contradiction, assume that there exists $\omega\in\Omega$ such that 
\begin{equation}\label{lemma_eq1}
\rho_{R^*}(\omega)< \rho_R(\omega).
\end{equation}
By \eqref{X_rhoA in A^*} and \eqref{A^**=A^*}, $X_{\rho_{R^*}}(\omega)\in (R^*)^*=R^*$. Then, \eqref{lemma_eq1} entails $X_{\rho_{R^*}}(\omega)\in R^r\setminus R$. 
This in turn implies the existence of $(t_n(\omega))_{n\in\N}$ in $\R_+$ such that $t_n>\rho_{R^*}$, $X_{t_n}\in R$, and $t_n\downarrow \rho_{R^*}$. It follows that $\rho_R=\lim_{n\to\infty} t_n=\rho_{R^*}$, a contradiction to \eqref{lemma_eq1}. 
With $\rho_R = \rho_{R^*}$ $\P^x${-a.s.} for all $x \in \R^d$, \eqref{lemma.finely2} and \eqref{lemma.finely2'} directly follow from \eqref{J} and \eqref{S, I, C}.
\end{proof}

\begin{Proposition}\label{prop.equiv.finely}
For any $R\in\B$,  
$R\in\Ec$ if and only if $R^*\in\Ec$. 
Moreover, if $R\in\Ec$, then any $T\in \B$ with $T^*=R^*$ belongs to $\Ec$ and satisfies
$J(x, T)=J(x, R)=J(x, R^*)$ for all $x\in \R^d.$
\end{Proposition}

\begin{proof}
Fix $R\in\B$. Suppose $R\in\Ec$, i.e. 
\begin{equation}\label{A equilibrium}
R= \Theta(R) =S(R)\cup(I(R)\cap R). 
\end{equation}
By \eqref{lemma.finely2'},  
\begin{equation}
\begin{aligned}
\Theta(R^*)= S(R^*)\cup(I(R^*)\cap R^*)&=S(R)\cup(I(R)\cap(R\cup R^r))\\
&=S(R)\cup(I(R)\cap R)\cup(I(R)\cap R^r)\\
&=R\cup R^r=R^*,
\end{aligned}
\end{equation}
where the third equality follows from \eqref{A equilibrium} and \eqref{rem:A^r in I}. This shows that $R^*$ also belongs to $\Ec$. Conversely, suppose $R^*\in\Ec$. Then \eqref{A equilibrium} holds with $R$ replaced by $R^*$, i.e. $R^*=S(R^*)\cup(I(R^*)\cap R^*)$. This can be rewritten, using \eqref{lemma.finely2'} and $R^*=R\cup R^r = R\cup (R^r\setminus R)$, as
\begin{align}\label{1_5}
R\cup(R^r\backslash R) &= S(R)\cup(I(R)\cap \left(R\cup(R^r \backslash R)\right))\notag\\
&=S(R)\cup \left( I(R)\cap R \right)\cup \left( I(R)\cap (R^r\backslash R) \right)\notag\\
&= S(R)\cup \left( I(R)\cap R\right)\cup (R^r\backslash R), 
\end{align}
where the last equality follows from \eqref{rem:A^r in I}. Note that \eqref{rem:A^r in I} also implies $S(R)\cap (R^r\backslash R)\subseteq S(R)\cap I(R)=\emptyset$. Hence, in \eqref{1_5}, the left hand side is a disjoint union of $R$ and $R^r\setminus R$, and the right hand side is a disjoint union of $S(R)\cup \left( I(R)\cap R\right)$ and $R^r\setminus R$. We then conclude from \eqref{1_5} that $R = S(R)\cup \left( I(R)\cap R \right)=\Theta(R)$, i.e. $R\in\Ec$. 

Now, if $R\in\Ec$, by part (i) $R^*\in\Ec$. For any $S\in\B$ with $S^*=R^*$, as $S^*=R^*\in\Ec$, we have $S\in\Ec$ (by part (a) again). Then, Lemma \ref{lemma.equiv} directly gives $J(x, S)=J(x,S^*)=J(x, R^*)=J(x, R)$, for all $x\in\R^d$.
\end{proof}

In view of Proposition~\ref{prop.equiv.finely}, to find an optimal equilibrium, it suffices to restrict our attention to finely closed stopping policies. After all, the fine closure of $R\in\Ec$ remains an equilibrium, with the same values. In fact, as Lemma~\ref{lemma.equiv} indicates, $R$ and $R^*$ induce the same stopping behavior, with $S(R)$, $I(R)$, and $C(R)$ in \eqref{S, I, C} staying intact after we take the fine closure of $R$.


\subsection{Euclidean Closures versus Fine Closures}

\begin{Definition}[\cite{blumenthal2007markov}, Definition II.3.1]\label{def.polar}
Given $R\in \B$, we say that $R$ is polar if $\rho_R=\infty$ $\P^x$-a.s. for all $x\in \R^d$, that $R$ is thin if $R^r=\emptyset$, and that $R$ is semipolar if it is a countable union of thin sets.
\end{Definition}

Instead of dealing with all stopping policies $R\in\B$, we focus on those such that 
\begin{equation}\label{difference polar}
\overline{R}\setminus R^*\ \text{is semipolar},
\end{equation}
where $\overline R$ denotes the (Euclidean) closure of $R$. 

\begin{Remark}\label{rem:include closed sets}
\eqref{difference polar} 
covers all closed subsets of $\R^d$. Indeed, if $R\in\B$ is closed, $\overline{R}\setminus R^* = R\setminus R=\emptyset$ is trivially semipolar. 
\end{Remark}

\begin{Remark}\label{rem:exclude Q}
\eqref{difference polar} excludes some pathological sets that are so small that $X$ will never reach, but so dense that their closures are immensely larger and will be hit by $X$ with positive probability. For instance, if $X$ is a $d$-dimensional Brownian motion with $d\ge 2$, then 
\[
Q:=\{x=(x_1,x_2,...,x_d)\in\R^d : x_i\in\Q,\ \text{$i=1,2,...,d$}\}
\]
is polar, but $\overline Q=\R^d$. Note that \eqref{difference polar} excludes $Q$. Since $Q$ is polar, $Q^r = \emptyset$ and thus $Q^* = Q\cup Q^r = Q$ is polar. Then, $\overline Q\setminus Q^* = \R^d\setminus Q$ will be hit by $X$ continuously over time, and is therefore not semipolar (in view of \cite[Proposition II.3.4]{blumenthal2007markov}). 

In practice, one does not usually take into account a stopping policy like $Q$, but simply consider $\emptyset$ (giving the same effect ``never stop'' as $Q$) or $\overline Q=\R^d$ (``stop immediately''). 
\end{Remark}

\begin{Remark}
In the one-dimensional case (i.e. $d=1$), \eqref{1-d condition} ensures $R^*=\overline R$ for all $R\in\B$, so that \eqref{difference polar} is trivially satisfied for all $R\in\B$. Hence, \eqref{difference polar} covers the one-dimensional setup in 
 \cite{huang2017optimal}, and can be viewed as the multi-dimensional counterpart of \eqref{1-d condition}. 
 \end{Remark}

Combining the focus on finely closed stopping policies, stipulated at the end of Section~\ref{subsec:finely closed}, with the additional requirement \eqref{difference polar}, we end up with the following collection of stopping policies:
\begin{equation}\label{Lc}
\Lc:=\{R\in \B: R=R^*\ \hbox{and}\ \overline{R}\setminus R\ \text{is semipolar}\}.
\end{equation}

\begin{Lemma}\label{lemma.set.0}
$\Lc$ contains all closed subsets of $\R^d$, and is closed under finite unions and countable intersections. 
\end{Lemma}

\begin{proof}
The first assertion simply follows from Remark~\ref{rem:include closed sets}. 
For any $R, T\in \Lc$, using the fact that $\overline{R\cup T}= \overline R\cup \overline T$, we get 
\bee
\overline{R\cup T}\setminus (R\cup T)=(\overline R\cup \overline T)\setminus (R\cup T)\subseteq (\overline R\setminus R)\cup (\overline T\setminus T).
\eee
As $\overline R\setminus R$ and $\overline T\setminus T$ are both semipolar, $\overline{R\cup T}\setminus (R\cup T)$ is semipolar, i.e. $R\cup T\in \Lc$. On the other hand, given any nonincreasing sequence $(R_n)_{n\in \N}$ in $\Lc$, set $R:=\bigcap_n R_n$. In view of Remark~\ref{rem:fine topology}, since $R_n$ is finely closed for all $n\in\N$, their intersection $R$ is also finely closed. Moreover, since
$\overline{R}\subseteq \bigcap_n \overline{R_n}$, 
\begin{equation*}
\overline{R}\setminus R= \overline{R}\setminus \bigg(\bigcap_{n} R_n\bigg) \subseteq  \bigg(\bigcap_n \overline{R_n}\bigg)\setminus \bigg(\bigcap_n R_n\bigg).
\end{equation*}
Given any point $x\in \big(\bigcap_n \overline{R_n}\big)\setminus (\bigcap_n R_n)$, $x$ is contained in every $\overline{R_n}$, and there exists $n_0\in\N$ such that $x\notin R_{n_0}$; hence, $x\in \overline{R_{n_0}}\setminus R_{n_0}$. The above inclusion relation therefore implies 
\begin{equation}\label{intersect polar}
\overline{R}\setminus R \subseteq  \bigcup_n (\overline{R_n}\setminus R_n).
\end{equation}
Since $\overline{R_n}\setminus R_n$ is semipolar for all $n\in \N$, the right hand side above, as a countable union of semipolar sets, is semipolar. Thus, $\overline{R}\setminus R$ is also semipolar, so that we can conclude $R\in \Lc$.
\end{proof}


Based on the development in this section, the appropriate subset of $\Ec$ we will focus on is 
\begin{equation}\label{tEc}
\widetilde\Ec := \Ec \cap \mathcal L.
\end{equation}


\section{Existence of an Optimal Equilibrium among $\tEc$}\label{sec:optimalE}
In this section, we set out to find an optimal equilibrium among $\tEc$ defined in \eqref{tEc}. We will first introduce a main assumption and its ramifications in Section~\ref{subsec:Hunt's}, and develop a machinery to improve equilibria in $\tEc$ in Section~\ref{subsec:improve}. The main result will be presented in Theorem~\ref{thm.optimal}.

\subsection{Hunt's hypothesis}\label{subsec:Hunt's}
By Definition~\ref{def.polar}, a polar set is clearly semipolar. The converse is the celebrated Hunt hypothesis, which is the second main assumption of this paper.

\begin{Assumption}[Hunt's hypothesis]\label{assum.hunt}
If $R\in\B$ is semipolar, then it is polar.
\end{Assumption}

Finding conditions which guarantee that a Markov process satisfies Assumption \ref{assum.hunt} is a classical topic in probabilistic potential theory. 
It is well-known that a $d$-dimensional Brownian motion satisfies Assumption \ref{assum.hunt} for all $d\in\N$. As a result, a large class of It\^{o} diffusions, given by \eqref{diffusion}, also satisfies Assumption \ref{assum.hunt}, as long as the Dol\'{e}ans-Dade exponential of $t\mapsto \int_0^t {b(X_s)}{\sigma^{-1}(X_s)}ds$ is a martingale, thanks to Girsanov's theorem;  
see e.g. \cite[Section 9.2]{oksendal2003stochastic}. 
For general Markov processes, we refer readers to \cite[Section 5.2]{chung2006markov} for a set of theoretic criteria that ensure Assumption \ref{assum.hunt}.

\begin{Remark}\label{rem:A-A^r polar}
For any $R\in \B$, $R\setminus R^r$ is semipolar; see Theorem 6 in \cite[Section 3.5]{chung2006markov}. Hence, under Assumption \ref{assum.hunt}, $R\setminus R^r$ is polar.
\end{Remark}


Assumption \ref{assum.hunt} leads to a very useful result in finding equilibria: if $R\in\B$ becomes smaller after we apply $\Theta$ to $R$ once, we immediately obtain an equilibrium. 

\begin{Proposition}\label{prop.decrease.theta}
Suppose Assumptions \ref{asm:reference measure} and \ref{assum.hunt} hold. Then, for any $R\in \B$ with $\Theta(R)\subseteq R$, $R\setminus \Theta(R)$ is polar and 
\[
\Theta^2(R)=\Theta(R),\quad \hbox{i.e.}\quad  \Theta(R)\in\Ec. 
\]
In addition, if $R$ is finely closed, so is $\Theta(R)$.
\end{Proposition}

\begin{proof}
For any $R\in\B$, by \eqref{A^r in I}, \eqref{Theta}, and $\Theta(R)\subseteq R$, we have 
\begin{equation}\label{eq.decrease.theta}
R^r\cap R\subseteq I(R)\cap R\subseteq \Theta(R)\subseteq R.
\end{equation}
As $R\setminus R^r$ is polar (Remark~\ref{rem:A-A^r polar}), this implies $R\setminus \Theta(R)$ is also polar. It follows that 
$\rho_{\Theta(R)}=\rho_R$ $\P^x$-a.s. for all $x\in \R^d$, which in turn implies $J(x, \Theta(R))=J(x, R)$  for all $x\in \R^d$. In view of \eqref{S, I, C}, we obtain $S(R)=S(\Theta(R))$ and $I(R)=I(\Theta(R))$. Hence,
\begin{align*}
\Theta^2(R) &= S(\Theta(R))\cup (I(\Theta(R))\cap \Theta(R)) \\
&= S(R) \cup (I(R)\cap \Theta(R))\\
&=S(R) \cup (I(R)\cap (S(R) \cup (I(R)\cap R) ))\\
&= S(R) \cup (I(R)\cap R) = \Theta(R),
\end{align*}
where the first, third, and fifth equalities follow from \eqref{Theta} and the fourth equality is due to $S(R)\cap I(R)=\emptyset$ by definition. This shows that $\Theta(R)\in\Ec$. Finally, as $\Theta(R)\subseteq R$, we have $(\Theta(R))^r\subseteq R^r$. If $R$ is additionally finely closed, i.e. $R^r\subseteq R$, then \eqref{eq.decrease.theta} yields $R^r\subseteq \Theta(R)$. It follows that $(\Theta(R))^r\subseteq R^r\subseteq \Theta(R)$, i.e. $\Theta(R)$ is finely closed. 
\end{proof}

\begin{Remark}
Under $\Theta(R)\subseteq R$, $R$ and $\Theta(R)$ can only differ in fairly limited ways. In view of \eqref{A^r in I}, every $x\in R\setminus \Theta(R)$ has to be a boundary point of $R$ that is not regular to $R$. 
Applying $\Theta$ to $R$ in this case removes those irregular boundary points, leading to the equilibrium $\Theta(R)$; see Remark~\ref{rem:by-product} for a concrete demonstration of this. 
\end{Remark}

\begin{Remark}\label{rem:convergence}
Proposition~\ref{prop.decrease.theta} enhances the convergence of the fixed-point iteration introduced in \cite{huang2018time}. When stated in the current context, \cite[Proposition 3.3]{huang2018time} asserts that whenever $R\subseteq \Theta(R)$, 
\[
\lim_{n\to\infty}\Theta^n(R) = \bigcup_{n\in\N} \Theta^n(R)
\] 
is well-defined and is an equilibrium. Proposition~\ref{prop.decrease.theta} complements the above result: for the opposite case $\Theta(R)\subseteq R$, $\lim_{n\to\infty}\Theta^n(R) = \Theta(R)$ is an equilibrium. 
\end{Remark}

\begin{Remark}\label{rem:Lc stopping}
Recall $\Lc$ in \eqref{Lc}. Under Assumption \ref{assum.hunt}, $\overline R\setminus R$ is polar for all $R\in\Lc$. Hence, for any $R\in\Lc$, $\rho_{R} = \rho_{\overline R}$ $\P^x$-a.s. for all $x\in\R^d$. 
\end{Remark}

\begin{Remark}\label{rem:to closedE}
It is tempting to conclude from Remark~\ref{rem:Lc stopping} that we can further restrict our attention from $\tEc$ to $\overline \Ec:= \{R\in\Ec: R\ \hbox{is closed}\}$; after all, the one-dimensional analysis in \cite{huang2017optimal} is entirely based on $\overline\Ec$. This is however not the case, as the relation ``$R\in\tEc$ if and only if $\overline R\in\Ec$'' does not hold in general. 
To illustrate, take any $R\in \tEc$. For $\overline R$ to be in $\Ec$, we need $f(x)\geq J(x, \overline R)$ for $x\in \overline{R}$. As $R\in\Ec$, we must have $f(x)\leq J(x, R)=J(x, \overline R)$ for $x\notin R$. Hence, ``$\overline R\in\Ec$'' boils down to the condition ``$f(x)=J(x, \overline R)$ for $x\in \overline{R}\setminus R$'', which is not true in general. From this observation, we construct an example in Section~\ref{sec:eg}, which explicitly demonstrates $R\in\tEc$ but $\overline R\notin \Ec$.  
%
\end{Remark}


\subsection{Improving an Equilibrium}\label{subsec:improve}

First, we observe that \cite[Lemma 3.1]{huang2017optimal} can be extended to the multi-dimensional case in a straightforward way.

\begin{Lemma}\label{lemma.smaller.eq}
For any $R$, $T\in \B$ with $R\subseteq T$ and $R\in\Ec$, $J(x, R)\geq J(x, T)$ for all $x\in \R^d$.
\end{Lemma}

\begin{proof}
Since the result follows from repeating the arguments in the proof of \cite[Lemma 3.1]{huang2017optimal}, we only sketch the proof below. As $R\subseteq T$, $\rho_T\le \rho_R$. For any $x\in\R^d$, observe that
\begin{align*}
J(x,R) &= \E^x[\delta(\rho_R) f(X_{\rho_R}) 1_{\{\rho_T < \rho_R\}}] + \E^x[\delta(\rho_R) f(X_{\rho_R}) 1_{\{\rho_T = \rho_R\}}] \\
&\ge \E^x[\delta(\rho_T)\E^x[ \delta(\rho_R-\rho_T) f(X_{\rho_R}) \mid \Fc_{\rho_T}] 1_{\{\rho_T < \rho_R\}}] + \E^x[\delta(\rho_T) f(X_{\rho_T}) 1_{\{\rho_T = \rho_R\}}] \\
&= \E^x[\delta(\rho_T)J(X_{\rho_T}, R)1_{\{\rho_T < \rho_R\}}] + \E^x[\delta(\rho_T) f(X_{\rho_T}) 1_{\{\rho_T = \rho_R\}}]\\
&\ge \E^x[\delta(\rho_T)f(X_{\rho_T})1_{\{\rho_T < \rho_R\}}] + \E^x[\delta(\rho_T) f(X_{\rho_T}) 1_{\{\rho_T = \rho_R\}}] = J(x,T),
\end{align*}
where the first inequality follows from \eqref{DI}, the second equality is due to the strong Markov property of $X$, and the second inequality stems from $R\in\Ec$ and $X_{\rho_T}\notin R$ on the set $\{\rho_T < \rho_R\}$. 
\end{proof}

The next result is a multi-dimensional extension of \cite[Proposition 4.8]{huang2017optimal}.

\begin{Proposition}\label{prop.intersection}
Assume Assumptions~\ref{asm:reference measure} and \ref{assum.hunt}, and that \eqref{assum.bound} and \eqref{old asm} hold for all $x\in\R^d$. 
Then, for any $R$, $T\in \tEc$,  
$\Theta( R\cap T)\subseteq R\cap T$ belongs to $\tEc$, and satisfies
\begin{equation}\label{eq.prop.intersection}
J(x, \Theta( R\cap T))\geq J(x, R)\vee J(x,  T), \quad \forall x \in \R^d.
\end{equation}
\end{Proposition}

\begin{proof}
Fix $R$, $T\in\tEc=\Ec\cap\Lc$. By the same arguments in the proof of \cite[Proposition 4.8]{huang2017optimal}, we get
\begin{equation}\label{from previous}
J(x, \overline R\cap \overline T)\geq J(x, \overline R)\vee J(x, \overline T), \quad \forall x \in (R\cap T)^c.
\end{equation}
As $R$, $T\in \Lc$, $\rho_R = \rho_{\overline R}$ and $\rho_T=\rho_{\overline T}$ $\P^x$-a.s. for all $x\in\R^d$ (Remark \ref{rem:Lc stopping}). It follows that 
\begin{equation}\label{JR=}
J(x,R)=J(x,\overline R)\quad \hbox{and}\quad J(x,T)=J(x,\overline T),\quad \hbox{$\forall x\in\R^d$}. 
\end{equation}
Moreover, by the same argument above \eqref{intersect polar}, we have $(\overline R\cap \overline T) \setminus (R\cap T)\subseteq(\overline R\setminus R) \cup (\overline T \setminus T)$. Since $\overline R\setminus R$ and $\overline T\setminus T$ are polar (Remark \ref{rem:Lc stopping}), so is $(\overline R\cap \overline T)\setminus (R\cap T)$. It follows that 
\begin{equation}\label{JRT=}
J(x,R\cap T) = J(x,\overline R\cap \overline T),\quad \forall x\in\R^d. 
\end{equation}
Now, by the fact $R,T\in\Ec$, we obtain from \eqref{JR=}, \eqref{from previous}, and \eqref{JRT=} that
\begin{equation}\label{eq.prop.intersection4}
f(x)\leq J(x, R)\vee J(x, T) =J(x, \overline R)\vee J(x, \overline T)\leq  J(x, \overline R\cap \overline T)= J(x,R\cap T), \quad \forall x\in (R\cap T)^c.
\end{equation}
This particularly implies $S(R\cap T) \subseteq R\cap T$, and thus
\begin{equation}\label{eq.prop.intersection3}
\Theta(R\cap T) = S(R\cap T) \cup (I(R\cap T)\cap(R\cap T))\subseteq R\cap T.
\end{equation}
By Proposition \ref{prop.decrease.theta}, this readily shows that $\Theta(R\cap T)\in\Ec$, $\Theta(R\cap T)$ is finely closed (as $R\cap T$ is finely closed), and
\begin{equation}\label{RT polar}
(R\cap T) \setminus \Theta(R\cap T)\ \hbox{is polar}. 
\end{equation}
Note that \eqref{eq.prop.intersection3} also implies
\bee
\Theta(R\cap T)\subseteq  \overline{\Theta(R\cap T)}\subseteq \overline R\cap \overline T .
\eee
It follows that
\begin{align*}
\overline{\Theta(R\cap T)} \setminus \Theta(R\cap T)  &\subseteq \big(\overline R\cap \overline T\big)\setminus \Theta(R\cap T)\\
 &\subseteq \Big((R\cap T) \setminus \Theta(R\cap T)\Big) \cup (\overline R\setminus R) \cup (\overline T\setminus T). 
\end{align*}
As the three sets in the second line above are all polar (recall \eqref{RT polar} and Remark \ref{rem:Lc stopping}), we conclude that $\overline{\Theta(R\cap T)} \setminus \Theta(R\cap T)$ is polar, and thus $\Theta( R\cap T) \in \Lc$. Hence, $\Theta( R\cap T)\in\Ec\cap\Lc=\tEc$. Finally, thanks to $\Theta(R\cap T) \subseteq R$ (by \eqref{eq.prop.intersection3}) and $\Theta(R\cap T)\in\Ec$, Lemma \ref{lemma.smaller.eq} asserts $J(x, \Theta(R\cap T))\geq J(x, R)$ for all $x \in \R^d$. A similar argument shows that $J(x, \Theta(R\cap T))\geq J(x, T)$ for all $x \in \R^d$. We can then conclude that \eqref{eq.prop.intersection} holds.  
\end{proof}

\begin{Remark}
In \eqref{from previous}, the inequality is guaranteed for only $x\in(R\cap T)^c$, although the corresponding one-dimensional result holds for all $x\in\R$; see \cite[Proposition 4.8]{huang2017optimal}. For instance, for $d\ge 2$, if there exist two closed equilibria $R$ and $T$ such that $R\cap T = \{x\}$ for some $x\in\R^d$, then $\rho_{R\cap T} = \rho_{\{x\}} =\infty$ $\P^x$-a.s., for a wide range of Markov processes $X$. By \eqref{old asm}, this implies $J(x,R\cap T)=0$, which is unlikely to be equal to $J(x,R)\vee J(x,T)$. By contrast, for $d=1$, \eqref{1-d condition} ensures $\rho_R=\rho_T=\rho_{R\cap T} = \rho_{\{x\}} =0$ $\P^x$-a.s., so that $J(x,R\cap T)= f(x)=J(x,R)\vee J(x,T)$.
\end{Remark}


Proposition \ref{prop.intersection} provides a partial order among elements in $\tEc$.

\begin{Corollary}\label{coro:partial order}
Assume Assumptions~\ref{asm:reference measure} and \ref{assum.hunt}, and that \eqref{assum.bound} and \eqref{old asm} holds for all $x\in\R^d$. Then, for any $R$, $T\in \tEc$ with $\overline R\subseteq \overline T$, 
$J(x, R)\geq J(x, T)$ for all $x\in \R^d$.
\end{Corollary}

\begin{proof}
For any $R$, $T\in \tEc$, note that 
$R \setminus \Theta({R}\cap{T}) \subseteq ({R}\setminus {T})\cup \big(({R}\cap {T}) \setminus\Theta({R}\cap {T})\big)$. 
With $\overline R\subseteq \overline T$, ${R}\setminus {T}\subseteq \overline{{T}} \setminus {T}$ is polar (Remark~\ref{rem:Lc stopping}). Recalling from \eqref{RT polar} that  $({R}\cap {T}) \setminus\Theta({R}\cap {T})$ is also polar, we conclude that ${R}\setminus \Theta({R}\cap {T})$ is polar. Hence, $\rho_{R}=\rho_{\Theta(R\cap T)}$ $\P^x$-a.s. for all $x\in\R^d$, and thus 
\bee
J(x, {R})=J(x, \Theta({R}\cap{T}))\ge J(x,T) \quad \forall x\in \R^d,
\eee
where the inequality follows from Proposition \ref{prop.intersection}.
\end{proof}


\subsection{The Main Result}
Before we state the main result of this paper, we need a convergence result for first hitting times. 

\begin{Lemma}\label{lemma.hitting.time}
Suppose Assumption  \ref{assum.hunt} holds. 
Let $(R_n)_{n\in\N}$ be a nonincreasing sequence in $\Lc$, and set $R:=\bigcap_{n\in \mathbb{N}}{R_n}$. 
Then
\begin{equation}\label{converge}
\lim\limits_{n\rightarrow \infty} \rho_{R_n}=\rho_{R}\ \ \P^x\hbox{-a.s.} \quad  \forall x\in R^c.
\end{equation}
\end{Lemma}

\begin{proof}
Fix $x\in R^c$. Set $\tau_n:=\rho_{R_n}$, and define $\tau:=\lim_{n\rightarrow \infty} \tau_n$. As $R\subseteq R_n$, we must have $\tau\leq \rho_R$. Hence, 
it suffices to prove 
\begin{equation}\label{tau>}
\tau\geq \rho_R\quad \P^x\text{-a.s. on}\  \  \{\tau <\infty\}.
\end{equation}
For each $m\in\N$, as $(R_n)_{n\in\N}$ is nonincreasing, $(\overline{R_n})_{n\in\N}$ is also nonincreasing. It follows that
\bee
X^x_{\tau_n}\in \overline{R_m}\quad \forall n\geq m,\quad  \P^x\hbox{-a.s. on}\ \ \{\tau <\infty\}.
\eee
As $n\to\infty$, by the continuity of $t\mapsto X^x_t$, this implies  
$X^x_\tau \in \overline{R_m}$ $\P^x$-a.s. on $\{\tau <\infty\}$.
Since $\overline{R_m}\setminus R_m$ is polar (recall $R_m\in \Lc$ and Remark~\ref{rem:Lc stopping}), the above relation can be equivalently written as
$X^x_\tau \in R_m$ $\P^x$-a.s. on $\{\tau <\infty\}$. By the arbitrariness of $m\in\N$,  we conclude 
\begin{equation}\label{in R}
X^x_\tau \in \bigcap_{m} R_m=R,\quad  \P^x\text{-a.s. on}\ \ \{\tau <\infty\}.
\end{equation}
As $x\in R^c$, $x\notin R_{n_0}$ for some $n_0\in\N$. Since $R_{n_0}$ is finely closed, $x\notin R_{n_0}$ implies that $x$ is not regular to $R_{n_0}$, i.e. $\tau_{n_0}=\rho_{R_{n_0}}>0$ $\P^x$-a.s. Consequently, $\tau>0$ $\P^x$-a.s. We then deduce from $\tau>0$ and \eqref{in R} that \eqref{tau>} holds.
\end{proof}

\begin{Remark}
We require ``$x\in R^c$'' in \eqref{converge}, as the convergence need not hold for $x\in R$. For instance, for any $d\ge 2$, let $X$ be a $d$-dimensional Brownian motion and $R_n\in\Lc$ be the closed ball around the origin $O:=(0,0,...,0)\in\R^d$ with radius $1/n$, for all $n\in\N$. Clearly, $R:=\bigcap_{n\in \mathbb{N}}{R_n}=\{O\}$. For $x=O\in R$, we have $\rho_{R_n} =0$ for all $n\in\N$, but $\rho_R=\infty$.   
\end{Remark}

Now, we are ready to state the main result of this paper.

\begin{Theorem}\label{thm.optimal}
Assume Assumptions \ref{asm:reference measure} and \ref{assum.hunt}, and that \eqref{assum.bound} and \eqref{old asm} hold for all $x\in\R^d$. Then, there exists ${R_\star}\in\tEc$ that is optimal among $\tEc$.  
Moreover, 
$\overline{R_\star} = \bigcap_{R\in \tEc}\Rb$. 
\end{Theorem}

\begin{proof}
Consider $\widetilde{R} := \bigcap_{R\in \tEc} \overline{R}$. As an intersection of closed sets, $\widetilde R$ is closed. Since the indicator function of a closed set is upper semi-continuous, \cite[Proposition 4.1 ]{bayraktar2012stochastic} implies that there exists a countable subset $({R_n})_{n\in\N}$ of $\tEc$ such that $\widetilde R=\bigcap_n \overline{R_n}$. Define $(T_n)_{n\in\N}$ by
\bee
T_1:=R_1,\quad T_n:=\Theta(T_{n-1}\cap R_n)\quad \hbox{for}\ n\ge 2.
\eee 
By applying Proposition \ref{prop.intersection} to $(T_n)_{n\in\N}$ recursively, we have 
\begin{align}\label{Tn in tEc}
T_n\in \tEc\quad \forall n\in\N,
\end{align}
as well as
\begin{align}\label{Tn in Rn}
T_{n+1}=\Theta(T_n\cap R_{n+1})&\subseteq T_n\cap R_{n+1}\notag\\
&\subseteq T_n = \Theta(T_{n-1}\cap R_n)\subseteq T_{n-1}\cap R_n\subseteq \overline{R_n},\quad \forall n\ge 2.
\end{align}
Consider $R_\circ:=\bigcap_n T_n$. We deduce from Lemma \ref{lemma.set.0}, \eqref{Tn in tEc}, and \eqref{Tn in Rn} that 
\begin{equation}\label{eq.thm.optimal2}
R_\circ\in \Lc\quad \hbox{and}\quad R_\circ\subseteq \bigcap_n \overline{R_n}=\widetilde{R}.
\end{equation}
Now, for any $x\in (R_\circ)^c$, as $(T_n)_{n\in\N}$ in $\tEc=\Ec\cap\Lc$ is nonincreasing (see \eqref{Tn in Rn}), Lemma \ref{lemma.hitting.time} entails $\rho_{T_n}\to \rho_{R_\circ}$ $\P^x$-a.s. Thanks to this, the continuity of $\delta$, $f$, and $t\mapsto X_t$, and \eqref{assum.bound}, we conclude from the dominated convergence theorem that 
\begin{equation}\label{eq.thm.optimal1}
\begin{aligned}
\lim_{n\rightarrow \infty} J(x, T_n)  = \lim_{n\rightarrow \infty}\E^x[\delta(\rho_{T_n}) f(X_{\rho_{T_n}})] = \E^x[ \delta(\rho_{R_\circ}) f(X_{\rho_{R_\circ}})]= J(x, R_\circ).
\end{aligned}
\end{equation}
On the other hand, the fact that $x\notin R_\circ=\bigcap_n T_n$ and $(T_n)_{n\in\N}$ is nonincreasing implies that there exists $n_0\in\N$ such that $x\notin T_{n}$ for all $n\ge n_0$. This, together with $T_n\in\Ec$ (by \eqref{Tn in tEc}), indicates 
$f(x)\leq J(x, T_{n})$ for all $n\ge n_0$. 
Combining this with \eqref{eq.thm.optimal1}, we obtain
\bee
f(x)\leq J(x, R_\circ)\quad \forall x\in (R_\circ)^c.
\eee
This shows that $S(R_\circ)\subseteq R_\circ$, so that 
\begin{equation}\label{dadu}
\Theta(R_\circ)= S(R_\circ)\cup (I(R_\circ)\cap R_\circ)\subseteq R_\circ. 
\end{equation}
By Proposition \ref{prop.decrease.theta}, this implies the following properties: 
\begin{equation}\label{polar again}
R_\circ \setminus \Theta(R_\circ)\quad \hbox{is polar};
\end{equation}
\begin{equation}\label{miss Lc}
{R_\star}:=\Theta(R_\circ)\ \hbox{belongs to $\Ec$ and is finely closed}.   
\end{equation}
Note that \eqref{dadu} implies $\Theta(R_\circ)\subseteq  \overline{\Theta(R_\circ)}\subseteq \overline{R_\circ}$, which gives
\begin{equation}\label{eq.thm.optimal0}
\overline{\Theta(R_\circ)} \setminus  \Theta(R_\circ) \subseteq  \overline{R_\circ} \setminus  \Theta(R_\circ) \subseteq (\overline{R_\circ} \setminus R_\circ)\cup ( R_\circ \setminus  \Theta(R_\circ)).
\end{equation}
As $R_\circ\in \Lc$ (by \eqref{eq.thm.optimal2}), $\overline{R_\circ} \setminus R_\circ$ is polar (recall Remark~\ref{rem:Lc stopping}). This, together with \eqref{polar again}, shows that the right hand side of \eqref{eq.thm.optimal0} is polar, and thus $\overline{\Theta(R_\circ)} \setminus  \Theta(R_\circ)$ is polar. We then conclude from \eqref{miss Lc} that ${R_\star}=\Theta(R_\circ)\in \tEc$. 
By \eqref{dadu} and \eqref{eq.thm.optimal2}, 
\begin{equation}\label{fR in tR}
{R_\star} \subseteq \widetilde R = \bigcap_{R\in\tEc} \overline{R}. 
\end{equation}
Hence, for any $R\in\tEc$, we have ${R_\star}\subseteq \overline R$. With ${R_\star}\in\Ec$, Corollary \ref{coro:partial order} gives $J(x, {R_\star})\ge J(x,R)$ for all $x\in \R^d$. Therefore, ${R_\star}$ is optimal among $\tEc$. Also, the fact ${R_\star}\in \tEc$ implies $\widetilde R=\bigcap_{R\in\tEc} \overline{R} \subseteq {\overline{R_\star}}$. This, together with \eqref{fR in tR}, entails ${\overline{R_\star}}=\widetilde{R}= \bigcap_{R\in \tEc} \overline{R}$.
\end{proof}


\subsection{An Illustration of the Use of Theorem~\ref{thm.optimal}}
Take $d=2$ and let $X=(X^{(1)}, X^{(2)})$ be a two-dimensional Brownian motion, which clearly satisfies Assumptions \ref{asm:reference measure} and \ref{assum.hunt}. We take up the hyperbolic discount function $\delta(t):=\frac{1}{1+\beta t}$ for some $\beta>0$, which satisfies \eqref{DI}. Given $a>0$, consider the payoff function $f:\R^2\to\R_+$ defined by 
\begin{equation}\label{f with a}
f(x_1,x_2) := |x_1-x_2|\wedge a,\quad \forall (x_1,x_2)\in\R^2.
\end{equation}
As $f$ is bounded, \eqref{assum.bound} and \eqref{old asm} hold trivially for all $x\in\R^2$. Now, for any $x\in\R^2$ and $R\in \B$, \eqref{J} takes the form
\begin{equation}\label{eq.eg1} 
J(x, R)=\E^x\bigg[\frac{|X^{(1)}_{\rho_R}-X^{(2)}_{\rho_R}| \wedge a}{1+\beta \rho_R}\bigg].
\end{equation}
This describes the situation where an agent expects the prices of two securities, $X^{(1)}_{t}$ and $X^{(2)}_t$, to diverge, and decides to take advantage of it via a long iron butterfly spread. Indeed, by taking $X^{(1)}$ as the underlying asset and $X^{(2)}$ as the (floating) target to compare with, the payoff of a long iron butterfly spread at time $t$ is exactly $f(X^{(1)}_t, X^{(2)}_t)$. Also, modeling security prices using Brownian motions, as in the Bachelier model, is justified when the price process may take negative values. In practice, CME Group (a mainstream financial derivatives exchange) announced in April 2020 the adoption of the Bachelier model to accommodate negative prices of certain energy derivatives.\footnote{See the announcement at \href{https://www.cmegroup.com/notices/clearing/2020/04/Chadv20-171.html}{https://www.cmegroup.com/notices/clearing/2020/04/Chadv20-171.html}.} 

To simplify \eqref{eq.eg1}, we change the coordinates by a rotation of $\pi/4$. Specifically, define 
\[
Y_t=\left[\begin{matrix}
Y^{(1)}_t\\
Y^{(2)}_t
\end{matrix}
\right]:=M \left[\begin{matrix}
X^{(1)}_t\\
X^{(2)}_t
\end{matrix}
\right],\quad \hbox{with}\ M:=
\left[
\begin{matrix}
\frac{1}{\sqrt{2}}&\frac{1}{\sqrt{2}}\\
-\frac{1}{\sqrt{2}}&\frac{1}{\sqrt{2}}
\end{matrix}
\right].
\]
In particular, 
\begin{equation}\label{Y_2}
Y^{(1)}_t := \frac{1}{\sqrt{2}}(X^{(1)}_2+X^{(2)}_t)\quad\hbox{and}\quad Y^{(2)}:=\frac{1}{\sqrt{2}}(X^{(2)}_t-X^{(1)}_t),  
\end{equation}
so that $Y$ remains a two-dimensional Brownian motion. For each $y\in\R^2$, similarly to the definition of $\P^x$ in the second paragraph of Section~\ref{sec.modelsetup}, we define $\P_Y^y$ as the law of $(Y^y_t)_{t\ge 0}$ and denote by $\E^y_Y$ the expectation taken under $\P^y_Y$. Note that
\[
\P^y_Y = \P^x\quad \hbox{for any $x, y\in\R^2$ with $y= Mx$}. 
\]
Moreover, for any $x,y\in\R^2$ and $R,T\in\B$ with $y=Mx$ and $T=MR$, we have
\[
\rho^{Y}_T:=\inf\{t>0: Y_t \in T\} = \inf\{t>0: X_t\in R\}=\rho_R,\quad \P^y_Y\hbox{-a.s. (or $\P^x$-a.s.)} 
\]
Hence, \eqref{eq.eg1} can be re-written as  
\be\label{eq.eg.equiv}  
J(x,R) = J^Y(y, T) := \sqrt{2}\cdot \E^y_Y \left[\frac{\big|Y^{(2)}_{\rho^Y_T}\big| \wedge \frac{a}{\sqrt{2}}}{1+\beta \rho^Y_T}\right], \quad \hbox{whenever}\ \ y = Mx,\ T=M R.
\ee
\begin{Remark}
Although the payoff at stopping in \eqref{eq.eg.equiv} relates only to the one-dimensional Brownian motion $Y^{(2)}$, $J^Y(y, T)$ remains a two-dimensional stopping problem. This is because the stopping time $\rho^Y_T$ still depends on a two-dimensional region $T$ in $\R^2$. 
\end{Remark}

Now, for any $0\le b\le a$, consider the region
\[
R_b := \{(x_1,x_2)\in\R^2 : x_1-x_2\ge  b\ \hbox{or}\ x_1-x_2\le - b\}\in \Lc.
\]

\begin{Lemma}\label{lemma.eg.BM}
Let $a^*>0$ be the unique solution to 
$a^*\int_0^\infty e^{-s}\sqrt{2\beta s}\tanh(a^*\sqrt{2\beta s})ds =1$.
Then, for any $0\le b\le a\wedge \sqrt{2}a^*$, $R_b\in \tEc$. 
\end{Lemma}

\begin{proof}
Fix $0\le b\le a\wedge \sqrt{2}a^*$ and define $T_b:=\{(y_1,y_2)\in\R^2: y_2\geq b/\sqrt{2}\ \hbox{or}\ y_2\leq -b/\sqrt{2}\}$. As $Y^{(2)}$ is a one-dimensional Brownian motion, when taking $T=T_b$ in \eqref{eq.eg.equiv}, we can apply the Bessel process analysis in \cite[Section 4]{huang2018time}. Specifically, with $b\le a$,
\[
J^Y(y,T_b) = \sqrt{2}\cdot \E^y_Y \Bigg[\frac{|Y^{(2)}_{\rho^Y_{T_b}}|}{1+\beta \rho^Y_{T_b}}\Bigg] > \sqrt{2}\cdot |y_2|,\quad \forall y\in\R^2\ \hbox{with}\ y_2\in (-b/\sqrt{2},b/\sqrt{2}),
\]
where the inequality is a direct consequence of \cite[Lemma 4.4]{huang2018time} and $b/\sqrt{2}\le a^*$; note that the characterization of $a^*$ is provided in \cite[Proposition 4.5]{huang2018time}.  Now, observe that $T_b = M R_b$. For any $x=(x_1,x_2)\in\R^2$ such that $x\notin R_b$, we have $y:= Mx\notin T_b$ (i.e., $y_2\in (-b/\sqrt{2},b/\sqrt{2})$). We then deduce from \eqref{eq.eg.equiv} and the previous inequality that 
\[
J(x,R_b) = J^Y(y,T_b)> \sqrt{2}\cdot |y_2| = |x_1-x_2| =  |x_1-x_2|\wedge a = f(x_1,x_2),  
\]
where the third equality follows from $|x_1-x_2|< b \le a$, thanks to $x\notin R_b$. This readily implies $R_b\in\Ec$. With $R_b\in\Lc$ by construction, we conclude $R_b\in\tEc$. 
\end{proof}

\begin{Proposition}\label{prop:R_a}
If $a\le \sqrt{2} a^*$, with $a^*>0$ as in Lemma~\ref{lemma.eg.BM}, then $R_{a}\in\tEc$ is optimal among $\tEc$. 
\end{Proposition}

\begin{proof}
By Theorem \ref{thm.optimal}, there exists $R_\star\in\tEc$ that is optimal among $\tEc$ and satisfies ${R_\star} \subseteq \bigcap_{R\in \tEc}\Rb$. As $a\le \sqrt{2} a^*$, Lemma \ref{lemma.eg.BM} implies $R_{a}\in \tilde{\Ec}$.  Since $R_a$ is closed by definition, we obtain $R_\star\subseteq R_a$. Suppose that $R_{a}\setminus R_\star\neq \emptyset$. For any $x=(x_1,x_2)\in R_{a}\setminus R_\star$, since $R_\star$ is finely closed, $\rho_{R_\star}>0$ $\P^x$-a.s. It follows that
\bee
J(x, R_\star)=\E^x\bigg[\frac{a}{1+\beta \rho_{R_\star}}\bigg]< a= |x_1-x_2|\wedge a = f(x),
\eee
where the first equality follows from $f\equiv a$ on $R_a\supseteq R_\star$ (recall \eqref{f with a}) and the second equality is due to $x\in R_a$. This shows that $R_\star$ is not an equilibrium, a contradiction. Hence, we must have $R_\star=R_a$, and thus $R_a$ is optimal among $\tEc$. 
\end{proof}

Proposition~\ref{prop:R_a} hinges on the one-dimensional analysis in Lemma~\ref{lemma.eg.BM} and the multi-dimensional result Theorem~\ref{thm.optimal}. While Lemma~\ref{lemma.eg.BM} already constructed a special collection of equilibria in $\R^2$ that are in a sense ``one-dimensional'' (as they only depends on the value $x_1-x_2$), it is unclear whether one should continue to study equilibria of more general forms in order to find an optimal equilibria. Theorem~\ref{thm.optimal} comes into play here, asserting that an optimal equilibrium exists and must be contained in the intersection of all the ``one-dimensional'' equilibria, which is $R_a$. Then, by the definition of $f$ in \eqref{f with a}, there is no way for an equilibrium to be properly contained in $R_a$,  leading to the only possibility that $R_a$ is an optimal equilibrium. 

\begin{Remark}
If $a> \sqrt{2} a^*$, an explicit characterization of an optimal equilibrium remains obscure. First, Lemma~\ref{lemma.eg.BM} now ensures  $R_{\sqrt{2} a^*}$, but not $R_a$, to be an equilibrium. Once we replace $R_a$ by $R_{\sqrt{2} a^*}$ in the proof of Proposition~\ref{prop:R_a}, we realize that the argument does not work anymore, as ``$f\equiv \sqrt{2} a^*$ on $R_{\sqrt{2} a^*}$'' fails to hold under the definition of $f$ in \eqref{f with a} and $a> \sqrt{2} a^*$. It is of interest as future research to investigate how an optimal equilibrium can be explicitly characterized for the case $a> \sqrt{2} a^*$. 
\end{Remark}


\section{An Example of $R\in \tEc$ but $\overline{R}\notin\Ec$}\label{sec:eg}
In this section, we take $X$ to be a three-dimensional Brownian motion, and will construct an example that explicitly demonstrates  $R\in\tEc$ but $\overline R\notin\Ec$. To this end, we need the following technical result, whose proof is relegated to Appendix~\ref{sec:proof of lemma.eg}.

\begin{Lemma}\label{lemma.eg}
Let $X$ be a three-dimensional Brownian motion. 
Given an open domain $D \subseteq \R^3$, suppose that $f\leq K$ 
on $\partial D$ for some $K>0$. Then, 
\bi
\item [(i)] 
for any $x\in D$ and $r>0$ such that $B(x, r):=\{y\in\R^3: \|y-x\|< r\}\subseteq D$, 
\begin{equation}\label{eq.eg.lemma0}
\frac{k(r)}{m(B(x, r))}\int_{B(x, r)} J(y, D^c) m(dy)\leq J(x, D^c)\leq \frac{1}{m(B(x, r))} \int_{B(x, r)} J(y, D^c) m(dy),
\end{equation}
where $k(r):=\E^x[\delta(\rho_{B(x, r)^c})]$ is continuous and nonincreasing in $r$ with $\lim_{r\downarrow 0}k(r)=1$, and $m(\cdot)$ denotes the Lebesgue measure in $\R^3$.
\item[(ii)] $x\mapsto J(x, D^c)$ is continuous on $D$. Furthermore, if $z\in\partial D$ is regular to $D^c$, then $x\mapsto J(x, D^c)$ is also continuous at $z$, in the sense that
\begin{equation}\label{conti at z}
\lim_{x\rightarrow z,\ x\in{D}} J(x, D^c)=J(z, D^c).
\end{equation}
\ei
\end{Lemma} 

Now, let $\mathcal S$ be the collection of $(x_1, x_2, x_3)\in \R^3$ such that
\bee
\begin{aligned}
x_1=\sqrt{x_2^2+x_3^2}, \quad &\text{for}\; x_1\geq 0,\ 0\leq x_2^2+x_3^2\leq 1;\\
x_1=2-\sqrt{x_2^2+x_3^2}, \quad &\text{for}\; x_1\geq 0,\ 1\leq x_2^2+x_3^2\leq 4;\\
x_2^2+x_3^2=4, \quad &\text{for}\; x_1<0.
\end{aligned}
\eee
As shown in Figure \ref{egpic}, $\mathcal S\subset \R^3$ is the surface generated by rotating the curve $l$ around the $x_1$-axis. It partitions $\R^3$ into two open sets $G_1$ and $G_2$: $G_1$ contains $(1, 0, 0 )$, $G_2$ contains $(-1,0,0)$, and $\partial G_1 =\partial G_2=\mathcal S$. 
Note that $\overline{G_1} = G_1\cup \mathcal S$ is a so-called {\it Lebesgue thorn} with the origin $O := (0,0,0)$ being its vertex. By \cite[Example 8.40]{peter2010}, 
\begin{equation}\label{O irregular}
\hbox{$O$ is not regular to either $G_1$ or $\overline{G_1}$, while all other points in $\mathcal S$ are regular to $G_1$}. 
\end{equation}

\begin{figure}[h!]
\centering
\includegraphics[width=11cm]{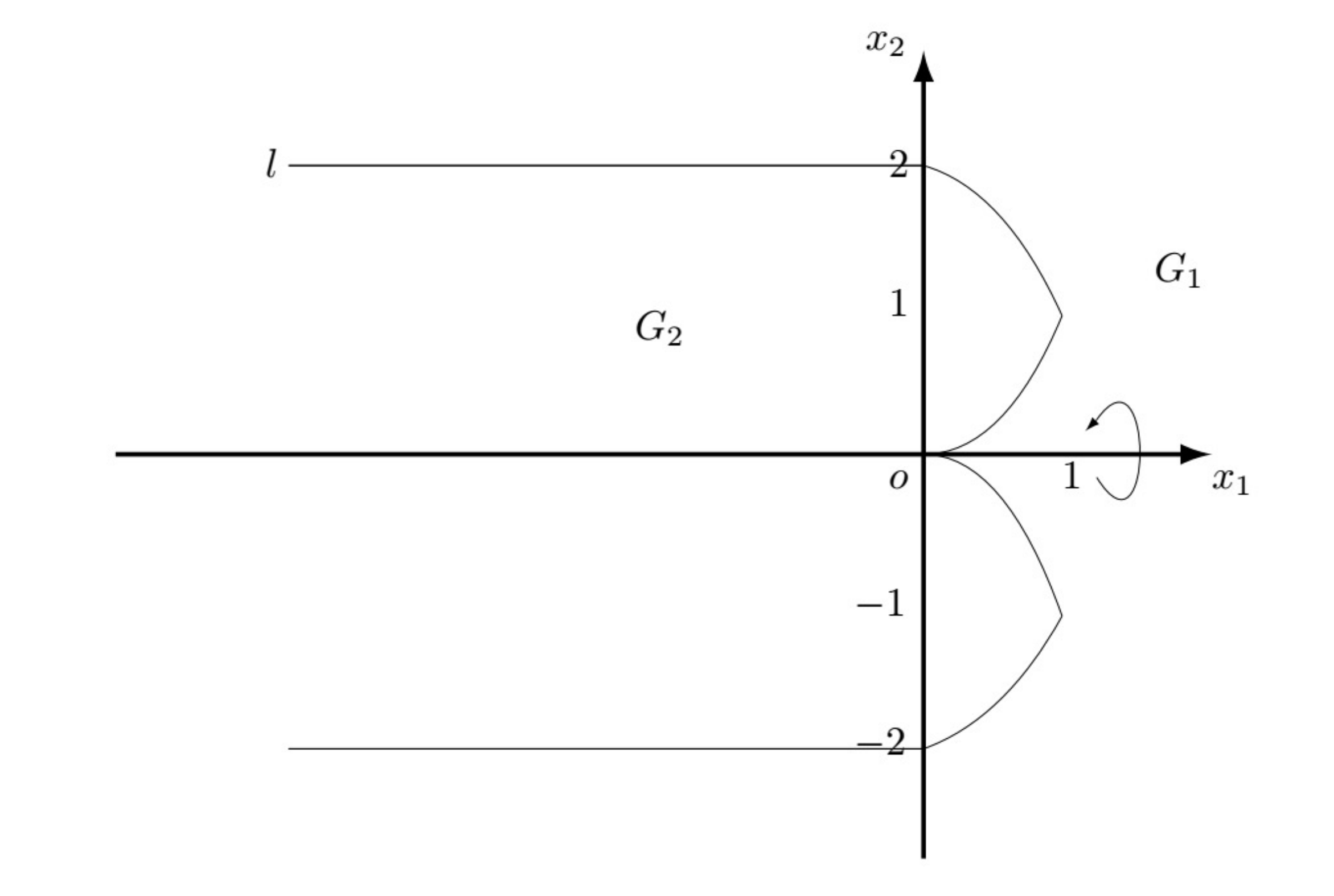}
\caption{Surface $S$ is generated by rotating the curve $\ell$ around the $x_1$-axis}\label{egpic}
\centering
\end{figure}
Define $h_1:\R^3\to\R_+$ by
\begin{align}
h_1(x):=
\begin{cases}
\sqrt{x_1^2+x_2^2+x_3^2},& \quad \text{for}\; x_1^2+x_2^2+x_3^2\leq 1,\\
1/\sqrt{x_1^2+x_2^2+x_3^2},& \quad \text{otherwise}.
\end{cases}
\end{align}
Note that $0<h_1\leq 1$ on $\mathcal S$ except at the origin $O$, where $h_1(O)=0$. Then, we introduce
\begin{equation}
h_2(x):=\E^x[\delta(\rho_{\overline{G_1}})h_1(X_{\rho_{\overline{G_1}}})],
\quad \forall x\in \R^3.
\end{equation}
Now, we define the payoff function $f$ by
\begin{align}\label{f eg}
f(x):=
\begin{cases}
 h_1(x), &\quad x\in \overline{G_1},\\
\min\{h_1(x), h_2(x)\}, &\quad x\in G_2.
\end{cases}
\end{align}

\begin{Lemma}
$f$ in \eqref{f eg} is continuous on $\R^3$, and satisfies \eqref{assum.bound} and \eqref{old asm} for all $x\in\R^3$.
\end{Lemma}

\begin{proof}
As $0\le h_1\le 1$ on $\R^3$, we have $0\le f\le 1$ on $\R^3$ by definition. Hence, \eqref{assum.bound} and \eqref{old asm} are trivially satisfied. Clearly, $h_1$ is continuous on $\R^3$. By \eqref{O irregular}, we conclude from Lemma \ref{lemma.eg} that $h_2$ is continuous on $\overline{G_2}\setminus \{O\}$. Moreover, since every point in $\mathcal S\setminus \{O\}$ is regular to $G_1$, by definition $h_2=h_1$ on $\mathcal S\setminus \{O\}$. All this readily implies that $f$ is continuous on $\R^3\setminus \{O\}$. Finally, observe that $0 \leq f(x) \leq h_1(x)$ for all $x\in \R^3\setminus\{O\}$.  This, along with $\lim_{x\to O} h_1(x) = h_1(O)=0$, entails $\lim_{x\to O} f(x) =0 = f(O)$. We therefore conclude that $f$ is continuous on $\R^3$.
\end{proof}

By \eqref{O irregular}, the fine closure of $G_1$ is 
\begin{equation}\label{G_1^*}
G_1^*=G_1\cup (\mathcal S\setminus \{O\}).
\end{equation}

\begin{Proposition}\label{prop.eg}
$G_1^*\in\tEc$ but $\overline{G_1^*}\notin\Ec$.
\end{Proposition}

\begin{proof}
By \eqref{O irregular} and $G_1$ being open, every point in $G_1^*$ is regular to $G_1^*$. Hence, 
$J(x, G_1^*)=\E^x[\delta(\rho_{G_1^*})f(X_{\rho_{G_1^*}})]=f(x)$ for all $x\in G_1^*$. 
For any $x\in(G_1^*)^c= G_2\cup\{O\}$, \eqref{O irregular} implies $X_{\rho_{G_1^*}} \in \mathcal S\setminus \{O\}$ $\P^x$-a.s. Hence, by the fact that $f=h_1$ on $\mathcal S$,
\begin{equation}\label{at G_2}
J(x,G_1^*)=\E^x[\delta(\rho_{G_1^*})f(X_{\rho_{G_1^*}})]= \E^x[\delta(\rho_{G_1^*})h_1(X_{\rho_{G_1^*}})]=h_2(x)\quad\forall x\in G_2\cup\{O\}, 
\end{equation}
where the last equality follows from the fact that $G_1^*$ and $\overline{G_1}$ only differ by the singleton $\{O\}$, which is polar (with respect to $X$, a three-dimensional Brownian motion). By the definition of $f$, this readily implies $J(x,G^*_1)=h_2(x)\ge f(x)$ for $x\in G_2$. Note that $X_{\rho_{G_1^*}} \in \mathcal S\setminus \{O\}$ also implies $h_1(X_{\rho_{G_1^*}})>0$. We then deduce from \eqref{at G_2} that $J(x,G^*_1)>0$ for all $x\in G_2\cup\{O\}$. In particular, 
\begin{equation}\label{at O}
J(O,G^*_1)> 0=h_1(O)=f(O).
\end{equation}
Therefore, we conclude that $G_1^*\subseteq I(G_1^*)$ and $S(G_1^*)=\emptyset$. In view of \eqref{Theta}, $\Theta(G_1^*)= G_1^*$, i.e. $G_1^*\in\Ec$. In addition, since $\{O\}$ is polar and $\overline{G^*_1} \setminus G_1^*=\overline{G_1} \setminus G_1^*= \{O\}$, we have $G_1^*\in \Lc$. Hence, $G_1^*\in\Ec\cap\Lc=\tEc$. Finally, by \eqref{O irregular} and \eqref{at O}, $J(O,\overline{G_1^*}) = J(O,{G_1^*}) >f(O)$, i.e. $O\in C(\overline{G_1^*})$. As $\overline{G_1^*} = G_1^*\cup\{O\}$ intersects $C(\overline{G_1^*})$, we conclude $\overline{G_1^*}\notin\Ec$. 
\end{proof}

\begin{Remark}\label{rem:by-product}
As a by-product, results in this section provide a concrete demonstration of Proposition~\ref{prop.decrease.theta}. In view of the proof of Proposition~\ref{prop.eg}, it can be checked directly that 
$
\Theta(\overline{G_1^*})={G_1^*}\subseteq \overline{G_1^*}.
$
By Proposition~\ref{prop.decrease.theta}, this implies $G_1^*=\Theta(\overline{G_1^*})\in\Ec$, which is verified by Proposition~\ref{prop.eg}. More specifically, $\overline{G_1^*}$ and $\Theta(\overline{G_1^*})$ only differ by the origin $O$ (which is not regular to $\overline{G_1^*}$), and it is the inclusion of $O$ that prevents  $\overline{G_1^*}$ from being an equilibrium. Applying $\Theta$ to $\overline{G_1^*}$ precisely removes this problematic point, turning $\overline{G_1^*}$ into the equilibrium $\Theta(\overline{G_1^*})= \overline{G_1^*}\setminus \{O\} = G_1^*$. 
\end{Remark}


    


\appendix

\section{Proof of Lemma~\ref{lem:density}}\label{sec:proof of lem:density}
Fix $R\in\B$. If $R$ is of zero potential, then 
\begin{equation}\label{U^alpha<}
 \int_0^\infty e^{-\alpha t}P_t(x, R)dt \le  \int_0^\infty P_t(x, R)dt= U(x,R)=0\quad \forall x\in\R^d.
\end{equation}
Hence, $\int_0^\infty e^{-\alpha t}P_t(x, R)dt=0$ for all $x\in\R^d$, and thus $\lambda^{\alpha}(R)=0$. Conversely, if $\lambda^{\alpha}(R)=0$, then 
\[
\lambda(E) =0\quad \hbox{with}\quad E:= \{x\in\R^d:U^\alpha(x,R) >0\},
\]
where we use the notation
\begin{equation}\label{U^alpha}
U^\alpha(x,R) := \int_0^\infty e^{-\alpha t} P_t(x,R)dt. 
\end{equation}
Since there exists a transition density $(p_t)_{t\ge 0}$ for $X$, it follows from \eqref{density} and $\lambda(E) =0$ that
\begin{equation*}
U(x,E)=\int_0^\infty P_t(x, E)dt=\int_0^\infty  \int_E p_t(x, y)\lambda(dy) dt =\int_0^\infty  0 dt=0\quad \forall x\in\R^d.
\end{equation*}
By the same argument as in \eqref{U^alpha<}, this implies
\begin{equation}\label{UE=0}
U^\alpha(x,E)= 0\quad \hbox{for all $x\in\R^d$}.
\end{equation} 
Now, we claim that $U^\alpha(x,R)= 0$ for all $x\in\R^d$. Let us write $R=R_1\cup R_2$, where
\[
R_1 := \{x\in R: U^\alpha(x,R)>0\} \quad\hbox{and}\quad R_2 := \{x\in R: U^\alpha(x,R)=0\}.
\]  
As $R_1\subseteq E$, we have $U^\alpha(x,R_1)\le U^\alpha(x,E)$ for all $x\in\R^d$. Then, \eqref{UE=0} entails $U^\alpha(x,R_1)=0$ for all $x\in\R^d$, so that 
\begin{equation}\label{A=A_2}
U^\alpha(x,R)=U^\alpha(x,R_2)\quad \forall x\in\R^d.
\end{equation}
By contradiction, suppose that there exists $x^*\in\R^d$ such that $U^\alpha(x^*,R)> 0$. Observe that
\begin{align*}
U^\alpha(x^*,R_2) &=\E^{x^*}\left[\int_0^\infty e^{-\alpha t} 1_{R_2}(X_t) dt\right] = \E^{x^*}\bigg[\int_{\rho_{R_2}}^\infty e^{-\alpha t} 1_{R_2}(X_t) dt\bigg]\\
&= \E^{x^*}\left[e^{-\alpha\rho_{R_2}} U^\alpha(X_{\rho_{R_2}},R_2) \right]  = \E^{x^*}\left[e^{-\alpha\rho_{R_2}} U^\alpha(X_{\rho_{R_2}},R) \right]=0,
\end{align*}
where the first equality stems from the definition of $U^\alpha$ in \eqref{U^alpha}, the fourth equality follows from \eqref{A=A_2},  and the last equality is due to $U^\alpha(X_{\rho_{R_2}},R)=0$ (by the definition of $R_2$). This contradicts  $U^\alpha(x^*,R_2)= U^\alpha(x^*,R)>0$. Hence, we have $U^\alpha(x,R)= 0$ for all $x\in\R^d$. Thanks again to Proposition 3 in \cite[Section 3.5]{chung2006markov}, this implies $U(x,R)= 0$ for all $x\in\R^d$, i.e. $R$ is of zero potential.


\section{Proof of Lemma~\ref{lemma.eg}}\label{sec:proof of lemma.eg}
(i) Fix $x\in D$. As $X$ is a three-dimensional Brownian motion and $\delta$ is continuous and nonincreasing, $r\mapsto k(r):=\E^x[\delta(\rho_{B(x, r)})]$ is continuous and nonincreasing in $r$, and $\rho_{B(x, r)}\downarrow 0$ as $r\downarrow0$ $\P^x$-a.s. By the dominated convergence theorem and $\delta(0)=1$, we get $\lim_{r\downarrow 0} k(r)=1$. Now, fix $r>0$ such that $B(x,r) \subseteq D$.
For any $0<s\le t$, by \eqref{DI} we get 
\begin{align}\label{eq.eg.lemma00}
J(x, D^c)&=\E^x[\delta(\rho_{D^c}) f(X_{\rho_{D^c}})]\notag\\
&\geq  \E^x[\delta(\rho_{B(x, s)^c}) \E^{x}[\delta(\rho_{D^c}-\rho_{B(x,s)^c}) f(X_{\rho_{D^c}}) \mid \Fc_{\rho_{B(x,s)^c}} ]  ]\notag\\
&= \E^x[\delta(\rho_{B(x, s)^c}) J(X_{\rho_{B(x,s)^c}}, D^c) ] .
\end{align}
As $X$ is a three-dimensional Brownian motion, $X^x_{\rho_{B(x,s)^c}}$ and $\rho_{B(x, s)^c}$ are independent, and $X^x_{\rho_{B(x,s)^c}}$ is uniformly distributed on $\partial B(x,s)$. Thus, \eqref{eq.eg.lemma00} implies
\begin{equation}\label{eq.eg.lemma1}
J(x, D^c)\geq \E^x[\delta(\rho_{B(x, s)^c})]\cdot\E^x[J(X_{\rho_{B(x,s)^c}}, D^c)]=k(s) \cdot \frac{1}{\Sigma(s)} \int_{\partial B(x, s)} J(y, D^c) \Sigma(dy),
\end{equation}
where $\Sigma(s)$ denotes the two-dimensional surface measure on $\partial B(x,s)$ in $\R^3$.
On the other hand, as $\delta$ is nonincreasing, 
\begin{equation}\label{eq.eg.lemma2}
J(x, D^c)\leq \E^x[\delta(\rho_{D^c}-\rho_{B(x,s)^c}) f(X_{\rho_{D^c}})] = \E^x[J(X_{\rho_{B(x,s)^c}}, D^c) ]= \frac{1}{\Sigma(s)} \int_{\partial B(x,s)} J(y, D^c) \Sigma(dy).
\end{equation}
Combining \eqref{eq.eg.lemma1} and \eqref{eq.eg.lemma2}, and using $k(r)\leq k(s)$ as $s\leq r$, we obtain 
\begin{equation*}
k(r) \cdot  \int_{\partial B(x, s)} J(y, D^c) \Sigma(dy)\leq \Sigma(s)J(x, D^c)\leq  \int_{\partial B(x, s)} J(y, D^c) \Sigma(dy), \quad \forall 0<s\leq r.
\end{equation*}
Integrating the above inequality with respect to $s$ from $0$ to $r$ yields
\begin{equation*}
k(r) \cdot\int_{B(x, r)} J(y, D^c) m(dy)\leq m(B(x, r))\cdot J(x, D^c)\leq \int_{ B(x, r)} J(y, D^c) m(dy),
\end{equation*}
which gives \eqref{eq.eg.lemma0}. 

(ii) First, we show that $x\mapsto J(x, D^c)$ is continuous on $D$. As $f\leq K$ on $\partial D$,  
\begin{equation}\label{<K}
J(x, D^c)\leq K \quad \forall x\in D.
\end{equation}
Fix an arbitrary $\varepsilon>0$. For any $x_1, x_2\in D$, choose $r>0$ small enough such that $B(x_i, r)\subseteq D$ for $i=1,2$ and $k(r)\geq 1-\varepsilon$. Then, by \eqref{eq.eg.lemma0},
\begin{equation*}
\frac{1-\varepsilon}{m(B(x_i, r))}\int_{B(x_i, r)} J(y, D^c) m(dy)\leq J(x_i, D^c)\leq \frac{1}{m(B(x_i, r))}\int_{B(x_i, r)} J(y, D^c) m(dy), \quad \text{for}\; i=1,2.
\end{equation*}
It follows that
\begin{align}\label{x1, x2}
J(x_1,D^c)&-J(x_2, D^c)\leq\frac{1}{m(B(x_1, r))} \int_{B(x_1, r)} J(y, D^c)m(dy)- \frac{1-\varepsilon}{m(B(x_2, r))} \int_{B(x_2, r)} J(y, D^c)m(dy)\notag\\
&\leq    \frac{1}{m(B(x_1, r))} \int_{B(x_1, r)\Delta B(x_2, r)} J(y, D^c)m(dy)+\frac{\varepsilon}{m(B(x_2, r))} \int_{B(x_2, r)} J(y, D^c) m(dy)\notag\\
&\leq  K\bigg( \frac{m(B(x_1, r)\Delta B(x_2, r))}{m(B(x_1, r))} +\varepsilon\bigg),
\end{align}
where $B(x_1, r)\Delta B(x_2, r)$ denotes the symmetric difference of $B(x_1, r)$ and $B(x_2, r)$, and the third equality is due to \eqref{<K}. By choosing $x_2$ sufficiently close to $x_1$, we can make $\frac{m(B(x_1, r)\Delta B(x_2, r))}{m(B(x_2, r))}\leq \varepsilon$, so that $J(x_1,D^c)-J(x_2, D^c)\leq 2K\varepsilon$. By switching $x_1$ and $x_2$ in \eqref{x1, x2}, we obtain the similar result that by choosing $x_2$ close enough to $x_1$, we get $J(x_2,D^c)-J(x_1, D^c)\leq 2K\varepsilon$. Hence, $|J(x_1,D^c)-J(x_2, D^c)|\leq 2K\varepsilon$ for $x_2$ sufficiently close to $x_1$. That is, $J(x, D^c)$ is continuous at $x_1$. By the arbitrariness of $x_1\in D$, $x\mapsto J(x,D^c)$ is continuous on $D$.

It remains to prove \eqref{conti at z}. Fix $z\in \partial D$ that is regular to $D^c$. By Proposition 1 in \cite[Section 4.4]{chung2006markov}, for any $\eta >0$, $x\mapsto \P^{x}( \rho_{D^c}\leq \eta)$ is lower seimicontinuous. Hence,
\bee
\liminf_{x\rightarrow z} \P^{x}( \rho_{D^c}\leq \eta)\geq \P^z(\rho_{D^c}\leq \eta)=1,
\eee
which implies
\be\label{eq.eg.lemmaii2}
\lim_{x\rightarrow z} \P^{x}(\rho_{D^c}\leq \eta)=1.
\ee
Given $r>0$, note that because $X$ is a Brownian motion,
\begin{equation}\label{constant P}
x\mapsto \P^x(\rho_{B(x, r)^c} >\eta)\quad \hbox{is constant}. 
\end{equation}
Observe that it holds for all $\eta>0$ that 
\be\label{eq.eg.lemmaii4}
\P^x(\rho_{D^c} < \rho_{B(x, r)^c} ) \ge \P^x(\rho_{D^c} \leq \eta < \rho_{B(x, r)^c} )\geq \P^x(\rho_{D^c}\leq \eta ) + \P^x(\rho_{B(x, r)^c}>\eta) -1.
\ee
By \eqref{eq.eg.lemmaii2} and \eqref{constant P}, this implies
\bee
\liminf_{x\rightarrow z} \P^x(\rho_{D^c} < \rho_{B(x, r)^c})\ge \P^z(\rho_{B(z, r)^c}>\eta),\quad \forall \eta>0. 
\eee
As $\lim_{\eta\downarrow 0} \P^z(\rho_{B(z, r)^c} > \eta)=1$, thanks again to $X$ being a Brownian motion, we conclude that  
\be\label{limP=1}
\lim_{x\rightarrow z} \P^x(\rho_{D^c} < \rho_{B(x, r)^c})=1.
\ee
On the set $\{\rho_{D^c} < \rho_{B(x, r)^c}\}$, we have $\|X^x_{\rho_{D^c}}-x\| < r$. Hence, for any $\varepsilon>0$, by the continuity of $f$, $\delta$, and $t\mapsto X_t$, we can choose $r>0$ small enough such that for all $x\in D$ with $\|x-z\|\leq r$, \begin{equation}\label{<eps}
|f(x)-f(z)|\leq \varepsilon,\ \ |f(X^x_{\rho_{D^c}})-f(x)| \leq \varepsilon\ \ \text{on}\; \{\rho_{D^c} < \rho_{B(x, r)^c}\},\ \ \hbox{and}\ \ \E^x[\delta(\rho_{B(x,r)^c})]\ge 1-\eps. 
\end{equation}
By \eqref{limP=1}, for this fixed $r>0$, we can choose $0<r'<r$ such that for all $x\in D$ with $\|x-z\|<r'$,
\begin{align}\label{use r'}
\P^x(\rho_{D^c} \ge \rho_{B(x, r)^c})\leq \varepsilon.
\end{align}
As $z$ is regular to $D^c$, $J(z, D^c)=f(z)$. If follows that
\be\label{eq.eg.lemmaii5} 
\begin{aligned}
|J(x, D^c)-J(z, D^c)| \leq \E^x\left[| \delta(\rho_{D^c})f(X_{\rho_{D^c}})-f(z) | \right].
\end{aligned}
\ee
Now, for any $x\in D$ with $\|x-z\|<r'$, observe that
\begin{align*}
&\E^x[|\delta(\rho_{D^c})\ f(X_{\rho_{D^c}})-f(z) |\ 1_{\{\rho_{D^c} < \rho_{B(x, r)^c}\}}]\\
 &\leq  \E^x[| \delta(\rho_{D^c})f(X_{\rho_{D^c}})-f(x) |\ 1_{\{\rho_{D^c} < \rho_{B(x, r)^c}\}}]+|f(x)-f(z)|\\
&\leq  \E^x[|f(X_{\rho_{D^c}})-f(x)|\ 1_{\{\rho_{D^c} \leq \rho_{B(x, r)^c}\}}]+\E^x[|\delta(\rho_{D^c})-1| |f(X_{\rho_{D^c}})|\ 1_{\{\rho_{D^c} \leq \rho_{B(x, r)}\}}]+\varepsilon\\
&\leq (K+2) \varepsilon,
\end{align*}
where the last inequality follows from \eqref{<eps} and $0\le f\le K$ on $\partial D$. On the other hand,
\bee
\E^x[| \delta(\rho_{D^c})f(X_{\rho_{D^c}})-f(z) |\ 1_{\{\rho_{D^c} \ge \rho_{B(x, r)^c}\}}]\leq 2K \P^x(\rho_{D^c} \ge \rho_{B(x, r)^c})\leq 2K \varepsilon,
\eee
where the last inequality follows from \eqref{use r'}. We then conclude from \eqref{eq.eg.lemmaii5} that $|J(x, D^c)-J(z, D^c)|\leq (3K+2) \varepsilon$, which completes the proof.

\bibliographystyle{plain}
\bibliography{reference}

\end{document}